\documentclass[twoside]{article}

\usepackage{fullpage}

\usepackage{xspace}
\usepackage[english]{babel}
\usepackage{amsmath}
\usepackage{amsopn}
\usepackage{latexsym}
\usepackage{textcomp}
\usepackage[font=small]{subcaption}
\usepackage[font=small]{caption}
\usepackage{color}

\usepackage{authblk}
\usepackage{enumerate}
\usepackage[inline]{enumitem}
\usepackage{tabularx}
\usepackage[plainpages=false]{hyperref}
\usepackage[figure,table]{hypcap}
\usepackage{multirow}
\usepackage{graphicx}
\usepackage{lineno}

\usepackage{amssymb,amsfonts,amsthm}
\usepackage{wrapfig}
\usepackage[all]{nowidow}

\graphicspath{{figures/}}

\newcommand{\cell}{\mathrm{cell}}

\newcommand{\myparagraph}[1]{\medskip\noindent\textbf{\boldmath #1}}
\newcommand{\core}{\circ}
\newcommand{\bbox}[2]{R(#1,#2)}

\newcommand{\convex}{convex\xspace}
\newcommand{\good}{good\xspace}
\newcommand{\sat}{resolved\xspace}

\newcounter{casecounter}
\newcounter{subcasecounter}
\newcounter{subsubcasecounter}
\makeatletter
\newcommand{\ccase}[2][]{\stepcounter{casecounter}\setcounter{subcasecounter}{0}\protected@write \@auxout {}{\string \newlabel {#2}{{#1\thecasecounter}{\thepage}{#1\thecasecounter}{#2}{}} }\hypertarget{#2}{\medskip\noindent\textbf{Case #1\thecasecounter.}}
}

\newcommand{\subcase}[2][]{\stepcounter{subcasecounter}\setcounter{subsubcasecounter}{0}\protected@write \@auxout {}{\string \newlabel {#2}{{#1\thecasecounter.\thesubcasecounter}{\thepage}{#1\thecasecounter.\thesubcasecounter}{#2}{}} }\hypertarget{#2}{\medskip\noindent\textbf{Case #1\thecasecounter.\thesubcasecounter.}}
}

\newcommand{\subsubcase}[2][]{\stepcounter{subsubcasecounter}\protected@write \@auxout {}{\string \newlabel {#2}{{#1\thecasecounter.\thesubcasecounter.\thesubsubcasecounter}{\thepage}{#1\thecasecounter.\thesubcasecounter.\thesubsubcasecounter}{#2}{}} }\hypertarget{#2}{\smallskip\noindent\textbf{Case #1\thecasecounter.\thesubcasecounter.\thesubsubcasecounter.}}
}
\makeatother
\newcommand{\newcase}{\setcounter{casecounter}{0}}

\newtheorem{theorem}{Theorem}
\newtheorem{lemma}{Lemma}

\newtheorem{property}{Property}
\newtheorem{corollary}{Corollary}

\begin{document}

\title{Greedy Rectilinear Drawings\thanks{This work started at the Bertinoro Workshop on Graph Drawing 2017, Italy. Research was partially supported by DFG grant Ka812/17-1 and by the project ``Algoritmi e sistemi di analisi visuale di reti complesse e di grandi dimensioni'' -- Ricerca di Base 2018, Dipartimento di Ingegneria dell'Universit\`a degli Studi di Perugia.}}

\author[1]{P.~Angelini}
\author[1]{M.~A. Bekos}
\author[2]{W.~Didimo}
\author[2]{L.~Grilli}
\author[3]{P.~Kindermann}
\author[4]{T.~Mchedlidze}
\author[4]{R.~Prutkin}
\author[5]{A.~Symvonis}
\author[2]{A.~Tappini}
\affil[1]{Institut f{\"u}r Informatik, Universit{\"a}t T{\"u}bingen, Germany\\
\texttt{$\{$angelini,bekos$\}$@informatik.uni-tuebingen.de}}
\affil[2]{Universit\`a degli Studi di Perugia, Italy,\newline
\texttt{$\{$walter.didimo,luca.grilli$\}$@unipg.it, alessandra.tappini@studenti.unipg.it}}
\affil[3]{Lehrstuhl f\"ur Informatik I, Universit\"at W\"urzburg, Germany\\
\texttt{philipp.kindermann@uni-wuerzburg.de}}
\affil[4]{Institute of Theoretical Informatics, Karlsruhe Institute of Technology, Germany\\
\texttt{mched@iti.uka.de, roman.prutkin@kit.edu}}  
\affil[5]{School of Applied Mathematical \& Physical Sciences, NTUA, Greece\\
\texttt{symvonis@math.ntua.gr}}

\maketitle

\begin{abstract}
	A drawing of a graph is \emph{greedy} if for each ordered pair of vertices $u$ and $v$, there is a path from $u$ to $v$ such that the Euclidean distance to $v$ decreases monotonically at every vertex of the path. From an application perspective, greedy drawings are especially relevant to support routing schemes in ad hoc wireless networks. The existence of greedy drawings has been widely studied under different topological and geometric constraints, such as planarity, face convexity, and drawing succinctness. We introduce \emph{greedy rectilinear drawings}, where edges are horizontal or vertical segments. These drawings have several properties that improve human readability and support network routing. 
	
	We address the problem of testing whether a planar \emph{rectilinear representation}, i.e., a plane graph with prescribed vertex angles, admits a greedy rectilinear drawing. We give a characterization, a linear-time testing algorithm, and a full generative scheme for \emph{universal} greedy rectilinear representations, i.e., those for which \emph{every} drawing is greedy. For general greedy rectilinear representations, we give a combinatorial characterization and, based on it, a polynomial-time testing and drawing algorithm for a meaningful subset of instances.
\end{abstract}

\section{Introduction}\label{se:introduction}

In a \emph{greedy drawing} of a graph in the plane every vertex is mapped to 
a distinct point and, for each ordered pair of vertices $u$ and $v$, there is 
a \emph{distance-decreasing} path from~$u$ to~$v$, i.e., a path such that the 
Euclidean distance to~$v$ decreases monotonically at every vertex of the path. 
Greedy drawings have been originally proposed to support \emph{greedy routing 
schemes} for ad hoc wireless 
networks~\cite{DBLP:journals/tcs/PapadimitriouR05,DBLP:conf/mobicom/RaoPSS03}. 
In such schemes, a node that has to send a packet to a destination~$v$ 
just forwards the packet to one of its neighbors that is closer to~$v$ than 
itself. In their seminal work, Papadimitriou and 
Ratajczak~\cite{DBLP:journals/tcs/PapadimitriouR05} 
showed that $3$-connected planar graphs form the largest class of graphs for which every 
instance may admit a greedy drawing, and they formulated two conjectures: 

\medskip\noindent\emph{Weak conjecture}: Every $3$-connected planar graph admits a greedy 
drawing.

\medskip\noindent\emph{Strong conjecture}: Every $3$-connected planar graph admits a 
\emph{convex} greedy drawing, i.e., a planar greedy drawing with convex 
faces.

\medskip Concerning the weak conjecture, Dhandapani~\cite{DBLP:journals/dcg/Dhandapani10} 
provided an existential proof for maximal 
planar graphs. Later on, Leighton and Moitra~\cite{DBLP:journals/dcg/LeightonM10} 
and Angelini et al.~\cite{DBLP:journals/jgaa/AngeliniFG10} independently 
settled the weak conjecture positively, by also describing constructive 
algorithms. Da Lozzo et al.~\cite{DBLP:conf/compgeom/LozzoDF17} strengthened 
these results, showing that in fact every $3$-connected planar graph admits a 
\emph{planar} greedy drawing, which may however contain non-convex faces. As 
such, this result sits in between the two conjectures, leaving the strong 
conjecture still open. For graphs that are not $3$-connected, N{\"{o}}llenburg 
and Prutkin~\cite{DBLP:journals/dcg/NollenburgP17} characterized the 
trees that admit a greedy drawing. Note that every greedy drawing of a tree 
is planar~\cite{DBLP:journals/networks/AngeliniBF12}.

Greedy drawings have also been investigated in terms of \emph{succinctness}, an important property that helps to make greedy routing schemes work in practice. A drawing is \emph{succinct} if the vertex coordinates are represented by a polylogarithmic number of bits. Since there exist greedy-drawable graphs in the Euclidean sense that do not admit a succinct greedy drawing~\cite{DBLP:journals/networks/AngeliniBF12}, several papers studied succinct greedy drawings in spaces different from the Euclidean one or according to a metric different from the Euclidean distance~\cite{DBLP:journals/tc/EppsteinG11,DBLP:conf/isaac/GoodrichS09,DBLP:journals/algorithmica/HeZ14,DBLP:journals/tcs/LeoneS16,DBLP:journals/tcs/WangH14}.  

A model related to greedy drawings is the one of \emph{self-approaching 
drawings}~\cite{DBLP:conf/gd/AlamdariCGLP12,DBLP:journals/jocg/NollenburgPR16}. 
A straight-line drawing is self-approaching if for any ordered pair of 
vertices $u$ and $v$, there is a path~$P$ from~$u$ to~$v$ in the drawing such that, for any 
point~$q$ on~$P$, as a point $p$ continuously moves along $P$ from $u$ to $q$, the Euclidean distance from $p$ to~$q$ always decreases.
Clearly, every self-approaching 
drawing is greedy, but not vice versa. Hence, self-approaching drawings are 
greedy drawings with stronger properties. In particular, their \emph{dilation}
is bounded by a constant~\cite{ikl-95}, while for greedy drawings it may be 
unbounded~\cite{DBLP:conf/gd/AlamdariCGLP12}. The dilation (or ``stretch-factor'') 
of a straight-line drawing is the maximum value of the ratio 
between the length of the shortest path between two vertices in the drawing 
and their Euclidean distance. 

\myparagraph{Motivation and Contribution.} The rich literature on greedy drawings described above witnesses the relevance of these kind of drawings both from a practical and from a theoretical perspective. In particular, our work enhances the research on greedy drawings that satisfy some interesting topological or geometric requirements, such as planarity~\cite{DBLP:conf/compgeom/LozzoDF17} 
and face convexity~\cite{DBLP:journals/algorithmica/HeZ14,DBLP:journals/tcs/PapadimitriouR05,DBLP:journals/tcs/WangH14}. 

We initiate the study of greedy drawings in the popular \emph{orthogonal drawing} 
convention~\cite{DBLP:books/ph/BattistaETT99,orth-handbook,DBLP:journals/siamcomp/Tamassia87}: 
Vertices are mapped to points 
and edges are sequences of horizontal and vertical segments (consequently, each vertex has 
degree at most $4$). More precisely, we introduce \emph{planar greedy 
rectilinear drawings}, i.e., crossing-free greedy drawings where each edge 
is either a horizontal or a vertical segment.
We address the following general question: ``Let~$H$ be a \emph{planar rectilinear 
representation}, i.e., a plane graph with given values ($90$, $180$, $270$ 
degrees) for the geometric angles around each vertex; is it possible to 
assign coordinates to the vertices of $H$ so that the resulting drawing is 
greedy rectilinear?''.

Figure~\ref{fi:greedy-repr-a} shows a rectilinear drawing that is not greedy; nonetheless, the 
corresponding rectilinear representation has a greedy drawing, as shown in 
Fig.~\ref{fi:greedy-repr-b}. Our question fits into the effective 
\emph{topology-shape-metrics} 
approach~\cite{DBLP:journals/tse/BatiniNT86,DBLP:journals/siamcomp/Tamassia87}, 
which first computes a planar embedding 
of the graph, then finds an embedding-preserving orthogonal representation, 
and finally assigns coordinates to vertices and bends to complete the 
drawing. The topology-shape-metrics approach is successfully used to compute graph layouts in several application domains, including information systems~\cite{DBLP:journals/tsmc/TamassiaBB88,DBLP:journals/spe/BattistaDPP02}, software design~\cite{DBLP:journals/ivs/EiglspergerGKKJLKMS04}, and computer networks~\cite{DBLP:journals/jgaa/CarmignaniBDMP02}; see also~\cite{orth-handbook}. 
We consider orthogonal drawings without bends and we address the last step of the topology-shape-metrics approach. Our contribution is as follows.
 
\begin{figure}[tb]
	\centering
	\subcaptionbox{\label{fi:greedy-repr-a}}{\includegraphics[page=1]{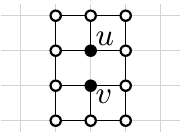}}
	\hfill
	\subcaptionbox{\label{fi:greedy-repr-b}}{\includegraphics[page=2]{greedy-repr}}
	\hfill
	\subcaptionbox{\label{fi:greedy-repr-c}}{\includegraphics[page=3]{greedy-repr}}
	\hfill
	\subcaptionbox{\label{fi:rectangular-inner-faces}}{\includegraphics[page=4]{greedy-repr}}
	\hfill
	\subcaptionbox{\label{fi:orthoconvex-outer-face}}{\includegraphics[page=5]{greedy-repr}}
	\caption{(a) A rectilinear drawing that is not greedy. (b) A greedy rectilinear drawing of the same representation (the distance-decreasing paths between $u$ and $v$ are dashed). (c) Drawing of a universal greedy rectilinear representation. (d)--(e) $H$ is not greedy realizable if an internal face is not a rectangle or the external face is not orthoconvex.}\label{fi:greedy-repr}
\end{figure}  

\begin{itemize}[label=--]
\item We discuss basic properties of greedy rectilinear drawings (Section~\ref{se:preliminaries}). In particular, we prove that the faces are always convex and the dilation is bounded by a small constant.
This makes these representations useful to improve human readability and support network routing.
In contrast, we provide convex (non-rectilinear) 
greedy drawings in which every distance-decreasing path between two vertices is arbitrarily longer than the Euclidean distance.

\item We investigate planar \emph{universal greedy} 
rectilinear representations, i.e., representations for which \emph{every} drawing is greedy (Section~\ref{se:universal}); see for example Fig.~\ref{fi:greedy-repr-c}). We give a characterization and a linear-time recognition algorithm which, 
in the positive case, computes a greedy drawing of minimum area on 
an integer grid. Our characterization may help in the design of algorithms that compute orthogonal representations in the second step of the topology-shape-metrics approach. 
We also describe a generative scheme for constructing 
any possible universal greedy rectilinear representation starting from a rectangle.

\item We extend our study to general rectilinear greedy representations (Section~\ref{se:greedy-rectilinear}). We give a non-geometric characterization, which leads to a linear-time testing algorithm for a meaningful subset of instances. If the condition of the characterization is satisfied, a greedy drawing of minimum area within that condition can be computed in quadratic time. However, we show that greedy rectilinear representations may require exponential area in general. Our non-geometric characterization opens up the way to intriguing theoretical problems, as discussed in Section~\ref{se:conclusions}.
\end{itemize}
We introduce basic concepts of graph drawing and the terminology used in the paper in 
Section~\ref{se:background}. Conclusions and open problems are reported in Section~\ref{se:conclusions}.

\medskip\noindent{\bf Methodological tools and strategy.} 
This paper mainly concentrates on $2$-connected graphs, because, as it will be shown in Section~\ref{se:preliminaries} (Theorem~\ref{th:trees}), the set of greedy rectilinear representations for $1$-connected graphs may be very limited. 
The main results of Section~\ref{se:universal} and Section~\ref{se:greedy-rectilinear} make use of two auxiliary planar DAGs (directed acyclic graphs) $D_x$ and $D_y$ associated with the input rectilinear representation $H$, which allow us to capture and summarize the relative $x$- and $y$-positions of pairs of vertices in a drawing of $H$. We prove that $H$ is universal greedy if and only if both $D_x$ and $D_y$ are Hamiltonian, which leads to an efficient linear-time testing algorithm for this family of representations (Theorem~\ref{th:universal-test}). More in general, the existence of a greedy drawing for $H$ depends on the existence of an $st$-ordering for each of $D_x$ and $D_y$ that guarantees specific connectivity properties for the subgraphs induced by consecutive nodes in that ordering (Theorem~\ref{thm:characterization}). This provides an interesting non-geometric characterization of greedy rectilinear representations and makes it possible to easily design a polynomial-time testing algorithm of greedy realizability for a large subclass of rectilinear representations, namely those for which $D_x$ and $D_y$ are series-parallel graphs. If the test is positive, then a greedy drawing of minimal area for the input representation can be found by solving a linear program (Theorem~\ref{thm:construction}). 

Concerning our generative scheme to create any universal greedy representation (Section~\ref{sse:generative-universal}), this is based on incrementally augmenting a rectilinear representation, starting from a rectangle and by using a small set of primitive operations that either subdivide edges of the external face or attach to the external face a simple path of reflex vertices. Path additions resemble those of an \emph{open ear decomposition} for $2$-connected graphs~\cite{robbins-39}, but they are tailored to plane graphs and of course enhanced with information concerned with the structure of a greedy rectilinear representation.

\section{Background}\label{se:background}

\myparagraph{Drawings and Planarity.} Let $G=(V,E)$ be a graph. A \emph{drawing} of $G$ is 
a geometric representation $\Gamma$ of $G$ in the plane such that each vertex 
$v \in V$ is placed at a distinct point $p_v$ and each edge $e=(u,v) \in E$ is 
drawn as a simple curve connecting $p_u$ and $p_v$. 
We denote by $x(v)$ and $y(v)$ the $x$- and the $y$-coordinate of a vertex $v \in V$ in $\Gamma$, respectively. 
For two vertices $u,v \in V$, we denote by $d(u,v)$ the Euclidean distance 
between $u$ and~$v$ in~$\Gamma$, and by $d_M(u,v)$ the Manhattan distance 
between them. Also, we say that a path from $u$ to $v$ in~$\Gamma$ is a 
\emph{u-v-path}. The degree of $v$ is denoted as $\deg(v)$. 

A drawing $\Gamma$ of a graph $G$ is \emph{planar} if no two edges intersect except at their common end-vertices (when they are adjacent). Graph $G$ is \emph{planar} if it admits a planar drawing $\Gamma$. Such a drawing divides the plane into topologically connected 
regions, called \emph{faces}. Exactly one face of $\Gamma$ is an unbounded 
region and it is called the \emph{external} face of $\Gamma$; the other 
faces are called \emph{internal}. Each internal face is described by the 
counterclockwise sequence of vertices and edges that form its boundary; while for the 
external face we use the clockwise sequence. 
The description of the set of (internal and external) faces 
determined by a planar drawing of $G$ is called a \emph{planar embedding} of~$G$. Recall that a planar embedding uniquely determines, for each vertex $v$, a clockwise ordering of 
the edges incident to $v$.
A planar graph~$G$ together with one of its planar embeddings is a 
\emph{plane graph}: If $\Gamma$ is a planar drawing of $G$ whose set of faces 
coincides with that described by the planar embedding of $G$, then~$\Gamma$ 
\emph{preserves} this embedding.

\myparagraph{Graph Connectivity.} A graph $G$ is \emph{$k$-connected} if 
every two vertices are connected by at least $k$ disjoint paths. If a graph 
is $k$-connected, for $k=1,2,3$, we also say that it is \emph{connected}, 
\emph{biconnected}, and \emph{triconnected}, respectively. 
Let~$G$ be a connected graph that is not biconnected. Then, $G$ contains at 
least a \emph{cutvertex}, namely a vertex whose removal disconnects $G$. 
Also, any maximal subgraph of $G$ that is biconnected is called a \emph{block}
 of~$G$. Finally, the \emph{block-cutvertex tree} $\mathcal{T}$ of $G$ is a 
tree whose \emph{C-nodes} are the cutvertices of $G$, and whose \emph{B-nodes}
are the blocks of $G$; then, $\mathcal{T}$ contains an edge between a B-node~$b$ 
and a C-node $c$ if and only if the cutvertex $c$ belongs to block $b$.

\myparagraph{Directed Graphs and Series-Parallel Compositions.}
A \emph{DAG} (directed acyclic graph) is a directed graph 
without directed cycles.  A node of a DAG with only outgoing (incoming) edges 
is a \emph{source} (\emph{sink}). A DAG $D$ is an $st$-\emph{digraph} if it has a 
single source $s$ and a single sink $t$. An \emph{$st$-ordering} of $D$ is a linear 
order $\mathcal{S} = v_1, \dots, v_n$ of its nodes such that $i < j$ for any directed edge 
$(v_i,v_j) \in D$; observe that $v_1=s$ and $v_n=t$ always holds. Every $st$-digraph $D$ 
admits an $st$-ordering, which can be computed in $O(n)$ time~\cite{EvenT76}. Finally, $D$ is \emph{series-parallel} if one of the following holds:
\begin{enumerate}[label=(\roman*)] 
	\item $D$ is a single edge $(s,t)$ connecting a source to a sink;
	\item $D$ is obtained from a set of series-parallel $st$-digraphs 
	$D_1,\dots, D_k$ with sources $s_1,\dots,s_k$ and sinks $t_1,\dots, t_k$, by 
	identifying $s_1,\dots,s_k$ into a single node $s$, which becomes the source 
	of~$D$, and $t_1,\dots,t_k$ into a single node~$t$, which 
	becomes the sink of~$D$. This operation is a \emph{parallel composition};
	\item  $D$ is obtained from a set of series-parallel $st$-digraphs 
	$D_1,\dots, D_k$ with sources $s_1,\dots,s_k$ and sinks $t_1,\dots, t_k$, by 
	identifying node $t_{i}$ with $s_{i+1}$, for each $i=1,\dots, k-1$. 
	Here, $s=s_1$ and $t=t_k$ are the source and the sink of the resulting graph $D$. This 
	operation is called \emph{series composition}.
\end{enumerate} 

\myparagraph{Orthogonal Drawings and Representations.} The concept of (rectilinear) orthogonal drawing has been already defined in the introduction. We now give a more formal definition of (rectilinear) orthogonal representations. Let $G$ be a plane graph, $v$ be a vertex of $G$, and $e_1, e_2$  
be two edges incident to $v$ that are consecutive in the clockwise order around $v$ (note that $e_1=e_2$, if $v$ has degree~$1$). We say that 
$a = \langle e_1,v,e_2 \rangle$ is an \emph{angle at $v$} of~$G$, or simply an \emph{angle} of 
$G$. Let $\Gamma$ and $\Gamma'$ be two rectilinear orthogonal drawings of~$G$ 
that preserve its planar embedding. We say that $\Gamma$ and $\Gamma'$ are 
\emph{shape equivalent} if for any angle $a$ of~$G$, the geometric angle 
corresponding to $a$ is the same in~$\Gamma$ and~$\Gamma'$.
In other words, two shape equivalent rectilinear orthogonal drawings~$\Gamma$ and~$\Gamma'$ may only differ for the coordinates of their vertices, while the angles around any vertex are the same in the two drawings.
Clearly, the shape equivalence relationship partitions the infinite set of rectilinear orthogonal drawings of a plane graph into a finite number of equivalence classes. Each of these classes is called a \emph{rectilinear orthogonal representation} $H$ of $G$. One can regard $H$ as a partial description of a drawing $\Gamma$ that only specifies the angles at each vertex but that does not fix the vertex coordinates. Hence, $H$ can be described by the embedding of~$G$ together with the geometric value of each angle of $G$ 
($90$, $180$, $270$ degrees)\footnote{Every degree-1 vertex has a single angle of 360 degrees, thus one can avoid to specify it.}.      
If $\Gamma$ is a rectilinear orthogonal drawing within class $H$, we  also say that $\Gamma$ is a rectilinear orthogonal drawing of $H$.
For example, Fig.~\ref{fi:greedy-repr-a} and Fig.~\ref{fi:greedy-repr-b} are two shape equivalent drawings, i.e., they are drawings of the same rectilinear orthogonal representation.

For the sake of simplicity, we will use the term \emph{rectilinear drawing} 
in place of rectilinear orthogonal drawing and the term \emph{rectilinear 
	representation} in place of rectilinear orthogonal representation.

Consider a rectilinear drawing $\Gamma$ of a rectilinear representation $H$. Since $H$ just fixes the angles around the vertices of $\Gamma$, rotating $\Gamma$ by a multiple of $90^\circ$ does not change $H$. Due to this observation, we can assume, without loss of generality, that $H$ always comes with a specific orientation of its edges, i.e., we shall assume that for every edge $(u,v)$ of $H$, it is fixed whether $u$ is to the left, to the right, above, or below $v$ in every rectilinear drawing $\Gamma$ of $H$.
A \emph{flat vertex} of $H$ (or of $\Gamma$) is a vertex with a flat angle ($180$ degrees). A flat angle formed by two horizontal segments is \emph{north-oriented} (\emph{south-oriented}) if it is above (below) the two segments. A flat angle between two vertical segments is either \emph{east-oriented} or \emph{west-oriented}. Finally, a \emph{staircase path} between two vertices $u$ and $v$ of $H$ (resp. of $\Gamma$) is either an edge or an $x,y$-monotone (zigzag) path that connects $u$ and $v$. 

\myparagraph{Greedy Drawings.} Let $\Gamma$ be a drawing of $G$. A path $(v_0, v_1, \dots, v_k)$ 
of $G$ is \emph{distance-decreasing} if $d(v_{i+1},v_k) < d(v_i,v_k)$, 
for $i=0, \dots, k-1$. Drawing~$\Gamma$ is \emph{greedy} if for any ordered pair 
of vertices $u,v$, there exists a distance-decreasing $u$-$v$-path.
If a rectilinear representation $H$ admits a greedy rectilinear drawing,~$H$ 
is \emph{greedy realizable} or, equivalently, it is a \emph{greedy rectilinear representation}.

\section{Basic Properties of Greedy Rectilinear Representations}\label{se:preliminaries}

In this section, we discuss some properties of rectilinear representations 
with respect to their possible greedy rectilinear drawings. We start with an 
additional definition concerning general (not necessarily rectilinear) greedy 
drawings. Let~$G$ be a graph, and let $v$ be a vertex of $G$ with 
neighbors $u_1, u_2, \dots, u_h$. The \emph{cell} of $v$ in a drawing $\Gamma$ of $G$, 
denoted by $\cell(v)$, is the (possibly unbounded) region of all points of 
the plane that are closer to $v$ than to any~$u_i$. The following geometric 
characterization is proven in~\cite{DBLP:journals/tcs/PapadimitriouR05}. 

\begin{theorem}[Papadimitriou and Ratajczak~\cite{DBLP:journals/tcs/PapadimitriouR05}]\label{th:papa-charact}
A drawing of a graph is greedy if and only if for every vertex $v$, $\cell(v)$ contains no vertex other than~$v$. 
\end{theorem}

For a rectilinear representation $H$, the cell of each vertex $v$ has a specific shape in any rectilinear drawing of $H$, which depends on $\deg(v)$ and on the angles at $v$. Fig.~\ref{fi:ortho-cells} shows all the possible shapes. Note that, if $\deg(v) \leq 3$, then $\cell(v)$ is always unbounded. 

\begin{figure}[t]
	\centering
	\subcaptionbox{}{\label{fi:cell-deg-4}\includegraphics[page=5]{cell}}
	\hfill
	\subcaptionbox{}{\label{fi:cell-deg-3}\includegraphics[page=4]{cell}}
	\hfill
	\subcaptionbox{}{\label{fi:cell-deg-2a}\includegraphics[page=2]{cell}}
	\hfill
	\subcaptionbox{}{\label{fi:cell-deg-2b}\includegraphics[page=3]{cell}}
	\hfill
	\subcaptionbox{}{\label{fi:cell-deg-1}\includegraphics[page=1]{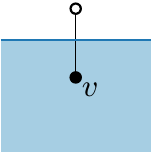}}
	\caption{Different types of cells (shaded regions) of a vertex $v$ in a rectilinear drawing of a graph: (a) $\deg(v)=4$; (b) $\deg(v)=3$; (c)-(d) $\deg(v)=2$, (e) $\deg(v)=1$.
	}\label{fi:ortho-cells}
\end{figure}

We restrict our study to biconnected graphs, because if a graph is not biconnected, the set of its greedy rectilinear drawings may be very limited, as shown by the following result for trees. 

\begin{theorem}\label{th:trees}
	A tree $T$ of vertex degree at most four admits a greedy rectilinear drawing 
	if and only if it has at most four leaves.
\end{theorem}
\begin{proof}
	Given a leaf $v$ of $T$ and a rectilinear drawing $\Gamma$ of $T$, we say 
	that $v$ is \emph{north-oriented} (\emph{south-oriented}) if $v$ is above (below) 
	its neighbor in $\Gamma$. Similarly,~$v$ is \emph{east-oriented} 
	(\emph{west-oriented}) if $v$ is to the right (left) of its neighbor in $\Gamma$. 
	If~$T$ has at least five leaves, then there are at least two leaves $u$ and~$v$ 
	in $\Gamma$ that are equally oriented, say north-oriented. This implies 
	that $\cell(u)$ contains~$v$ or $\cell(v)$ contains $u$ (or both). By 
	Theorem~\ref{th:papa-charact}, $\Gamma$ is not~greedy.
	
	\begin{figure}[tb]
		\centering
		\subcaptionbox{\label{fi:tree-a}}{\includegraphics{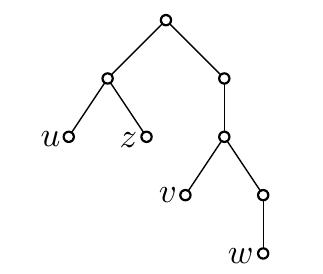}}
		\hfil
		\subcaptionbox{\label{fi:tree-b}}{\includegraphics{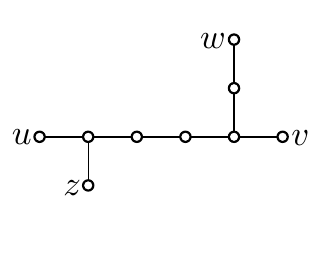}}
		\caption{(a) A tree $T$ with four leaves. (b) A greedy rectilinear representation of $T$}
		\label{fi:tree}
	\end{figure}
	
	Suppose vice versa that $T$ has at most four leaves $u,v,w,z$. A greedy 
	drawing~$\Gamma$ of $T$ is constructed as follows (see Figs.~\ref{fi:tree-a} 
	and~\ref{fi:tree-b}). All vertices of the path~$\pi$ between $u$ and $v$ in $T$ 
	are horizontally aligned, so that $u$ is west-oriented and $v$ is east-oriented. 
	Let $\pi'$ be the path connecting $\pi$ to $w$ and $\pi''$ be the 
	path connecting $\pi$ to $z$; note that $\pi'$ and $\pi''$ may be attached to 
  the same vertex of~$\pi$. Then, all vertices of $\pi'$ are vertically aligned 
	in such a way that~$w$ is north-oriented, while all vertices of $\pi''$ are 
	vertically aligned in such a way that~$z$ is south-oriented. It is immediate 
	to see that $\Gamma$ is a greedy drawing. 
\end{proof}

Observe that, with the same argument as in the proof of Theorem~\ref{th:trees}, it is possible to prove that any graph with more than four degree-$1$ vertices does not admit a greedy rectilinear drawing.

\subsection{Convexity of greedy rectilinear representations}
\medskip
We now show that we can further restrict our study to \emph{\convex} rectilinear 
representations, i.e., those whose every internal face is rectangular and whose external face is bounded by an \emph{orthoconvex} polygon. Recall that a simple polygon $P$ is orthoconvex if for any vertical or horizontal line $\ell$, the intersection between $P$ and $\ell$ is either empty or a single segment. An illustration of a \convex rectilinear representation is given in Fig.~\ref{fi:orthoconvex-example}.

\begin{lemma}\label{le:necessityconvex}
A rectilinear representation is greedy realizable only if it is convex.
\end{lemma}
\begin{proof}
	Suppose first that a rectilinear representation $H$ has an internal face~$f$ 
  that is not a rectangle. This means that there is a vertex $v$ with an angle 
  of~$270$ degrees inside $f$. 
	Let $\Gamma$ be any rectilinear drawing of~$H$. Suppose that, when moving 
  counterclockwise along the boundary of~$f$, we enter $v$ horizontally from 
  west and leave $v$ vertically towards south; the other cases are symmetric. 
  Since~$\Gamma$ has no bend, there exists a vertex $u$ to the right 
	of $v$ and above $v$ (see Fig.~\ref{fi:rectangular-inner-faces-app}). 
	Therefore, $\cell(v)$ contains $u$, which means that $\Gamma$ is not greedy 
	by Theorem~\ref{th:papa-charact}.
	
	\begin{figure}[tb]
		\centering
		\subcaptionbox{\label{fi:orthoconvex-example}}{\includegraphics[page=4]{orthoconvex-outer-face}}
		\hfil
		\subcaptionbox{\label{fi:rectangular-inner-faces-app}}{\includegraphics{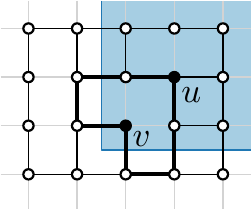}}
		\hfil
		\subcaptionbox{\label{fi:orthoconvex-outer-face-app}}{\includegraphics[page=3]{orthoconvex-outer-face}}
		\caption{
			(a) A \convex rectilinear representation, where the vertices on the orthoconvex outer face are filled white. 
			(b-c) Illustration for the proof that a rectilinear representation is not greedy realizable if (b) an internal face is not a rectangle, or 
			(c) the external face is not orthoconvex.}\label{fi:greedy-repr-app}
	\end{figure} 
	
	Suppose now that the polygon $P$ defined by the external boundary of $H$ is 
	not orthoconvex. Let $\Gamma$ be any rectilinear drawing 
	of $H$ and let $P$ be the polygon defined by the boundary of the external 
	face of $\Gamma$. Since $P$ is not orthoconvex, we can assume without loss of 
	generality that there exists a vertical line $\ell$ whose intersection 
	with~$P$ consists of at least two segments~$s'$ and $s''$. Suppose that~$s'$ is 
	above $s''$. Clearly, $\ell$ cuts $P$ into at least three distinct polygons, 
	two of which lie on the same side of $\ell$, say to the right of $\ell$: The 
	polygon having $s'$ as a leftmost side is denoted by~$P'$ while the polygon 
	having $s''$ as a leftmost side is denoted by~$P''$. Refer to 
	Fig.~\ref{fi:orthoconvex-outer-face-app} for an illustration. Let $r'$ be a rightmost 
	vertical side of~$P'$ and let $r''$ be a rightmost vertical side of $P''$; 
	also, denote by~$x(r')$ and~$x(r'')$ the $x$-coordinate of $r'$ and of $r''$, respectively.
	Assume first that $x(r') \geq x(r'')$ and let $p$ be the topmost point 
	on $r''$. Clearly,~$\Gamma$ has a vertex~$v$ at point $p$, and $v$ is a degree-$2$ 
	vertex that forms an angle of $270$ degrees on the external face of $\Gamma$. 
	Also, since $P'$ and $P''$ cannot intersect, there must be at least a vertex $u$ 
	of $\Gamma$ on $P'$ that is above $v$ and not to the left of $v$. It 
	follows that~$\cell(v)$ contains~$u$, and hence $\Gamma$ is not greedy. The 
	case in which $x(r') < x(r'')$ is handled symmetrically, choosing $p$ as the 
	bottommost point on~$r'$.
\end{proof}
 
For a rectilinear drawing of a \convex rectilinear representation $H$, and for any two vertices $u$ and $v$ of $H$, let $R(u,v)$ denote the minimum bounding box (rectangle or segment) including~$u$ and $v$. The next property immediately follows from the convexity of $H$.

\begin{property}\label{pr:right-angles}
Let $f$ be a face of $H$ and $w$ be any vertex of $H$ with an angle of $90$ degrees inside $f$. Denote by $u$ and $v$ the two neighbors of $w$ along the boundary of $f$. In any rectilinear drawing of $H$, there is no vertex properly inside $R(u,v)$.
\end{property}

\subsection{Dilation of greedy rectilinear representations} 
\medskip
We now exploit Property~\ref{pr:right-angles} to show that the dilation of greedy rectilinear representations is always bounded by a small constant. 

\begin{theorem}\label{th:bounded-dilation}
	In a rectilinear greedy drawing on an integer grid, for every two vertices $s,t$ there is a distance-decreasing $s$-$t$-path of length at most~$3 \sqrt{2} \cdot d(s,t)$.
\end{theorem}
\begin{proof}
	Let $\Gamma$ be a rectilinear greedy drawing of a rectilinear representation $H$. We prove that for every two vertices $s,t \in H$ there exists a distance-decreasing $s$-$t$-path in $\Gamma$ of length
	at most~$3 d_M(s,t)$. Then, the statement follows by the fact that the Manhattan distance between two points is at most $\sqrt{2}$ times the Euclidean distance between them, that is, $d_M(s,t) \leq \sqrt{2} d(s,t)$. 
	
	We use induction on $d_M(s,t)$, which is always an integer number, since $\Gamma$ has integer vertex coordinates. First note that every vertex $u$ of $H$ is connected to every vertex that is closest to $u$ in~$\Gamma$ with respect to the Euclidean distance~\cite{DBLP:journals/tcs/PapadimitriouR05}.
	
	In the base case $d_M(s,t)=1$, we have that $x(s)=x(t)$ or~$y(s)=y(t)$. 
	Since~$t$ is the closest vertex to~$s$, we have that~$s$ and~$t$ are adjacent 
	in~$\Gamma$, and the statement trivially holds. 
	
	Suppose now that $d_M(s,t)>1$ and that the statement holds for every pair of vertices whose Manhattan distance is less than $d_M(s,t)$. 
	If $x(s)=x(t)$ or $y(s)=y(t)$, then there must be in $\Gamma$ a distance-decreasing straight $s$-$t$-path (horizontal or vertical), as otherwise $\Gamma$ would not be greedy. In this case, the length of this path equals $d_M(s,t)=d(s,t)$. 
	Suppose now that $s$ and $t$ are not horizontally or vertically aligned. Without loss of generality, let $t$ lie to the right of $s$ and below it. Recall that~$\bbox{s}{t}$ denotes the bounding box of $s$ and $t$.
	We distinguish between the following cases:
	
	\newcase
	\ccase{c:dilation-boundary} There is a vertex~$v \neq s$ on
	the top or left boundary of~$\bbox{s}{t}$. Then,~$s$ is connected
	to~$v$ by a straight path. Since $d_M(v,t) < d_M(s,t)$, 
	by induction there exists a distance-decreasing $v$-$t$-path of length
	at most~$3 d_M(v,t)$. Concatenating this $v$-$t$-path and the
	straight $s$-$v$-path creates a distance-decreasing $s$-$t$-path of
	length at most \[d_M(s,v) + 3d_M(v,t) < 3 (d_M(s,v) + d_M(v,t)) = 3 d_M(s,t).\]
	
	\ccase{c:dilation-noboundary} There is no vertex~$v \neq s$ on
	the top or left boundary of~$\bbox{s}{t}$. Consider the shortest
	distance-decreasing $s$-$t$-path~$\rho$. Let $(s,u)$ be the first edge of 
	this path, and assume, without loss of generality, that this edge is 
	horizontal (the other case is symmetric). Then, we 
	have \[x(t) < x(u) < x(s) + 2(x(t)-x(s))\] and hence $d_M(u,t) < d_M(s,t)$.
	Note that, in this case, $s$ cannot have a neighbor~$w$ below it, as this would imply $y(w) < y(t)$, thus 
	violating Property~\ref{pr:right-angles}; see Fig.~\ref{fig:dilation:1}. 
	Then, the region delimited by the vertical lines $x = x(s)$,
	$x = \frac{1}{2}(x(s) + x(u))$, and by the upper horizontal line $y = y(s)$ 
	is a subset of $\cell(s)$, and thus it contains no vertex other than~$s$, since $\Gamma$ is greedy; see Fig.~\ref{fig:dilation:2}. Let~$v$ be 
	the latest successor of $u$ along $\rho$, such that the $u$-$v$-subpath of 
	$\rho$ is a staircase. The following subcases are possible:

	\begin{figure}[tb]
		\subcaptionbox{\label{fig:dilation:1}}{\includegraphics[page=1]{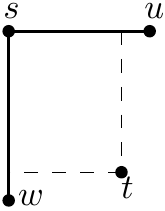}}
		\hfill
		\subcaptionbox{\label{fig:dilation:2}}{\includegraphics[page=2]{dilation}}
		\hfill
		\subcaptionbox{\label{fig:dilation:3}}{\includegraphics[page=3]{dilation}}
		\hfill
		\subcaptionbox{\label{fig:dilation:4}}{~\includegraphics[page=4]{dilation}}
		\caption{Illustration for the proof of Theorem~\ref{th:bounded-dilation}.
			(a)~There is no edge~$(s,w)$ with~$y(w)<y(s)$, and
			(b)--(d)~the path from~$u$ to~$v$ is the longest staircase on~$\rho$.}
		\label{fig:dilation}
	\end{figure}
	
	\subcase{sc:dilation-v=t} $v = t$. Then, the
	length of~$\rho$ is \[x(u)-x(s) + d_M(u,t) \leq 2 d_M(s,t) + d_M(u,t) \leq 3 d_M(s,t).\]
	
	\subcase{sc:dilation-v>t} $v \neq t$ and $y(v) > y(t)$. Then,~$x(v) > \frac{1}{2}(x(s) + x(u))$. Furthermore,
	$x(v) < x(t)$, as otherwise the edge following~$v$ would not be distance decreasing; see Fig.~\ref{fig:dilation:2}. The $s$-$v$-subpath
	of~$\rho$ has length \[x(u) - x(s) + x(u) - x(v) + y(s) - y(v) < 3\left(x(v)-x(s)\right)+y(s)-y(v)< 3d_M(s,v).\]
	Since $d_M(v,t) < d_M(s,t)$ by induction, the $v$-$t$-subpath
	of~$\rho$ has length at most $3 d_M(v,t)$.  Therefore, the
	length of~$\rho$ is at most~$3 d_M(s,v) + 3 d_M(v,t) = 3 d_M(s,t)$.
	
	\subcase{sc:dilation-v<t} $v \neq t$ and $y(v) \leq y(t)$. If $x(v) < x(t)$, then for the predecessor~$v'$
	of~$v$ on~$\rho$, we have $y(v') > y(t)$; see Fig.~\ref{fig:dilation:3}. Then, Case~\ref{sc:dilation-v>t} can be applied to~$v'$
	instead of~$v$. Conversely, assume that $x(v) \geq x(t)$; see
	Fig.~\ref{fig:dilation:4}. In this case, starting from $t$, we repeatedly go
	upwards or to the right in~$\Gamma$. We cannot get stuck, since
	otherwise we have a vertex with no edge to the right and no edge
	upwards, a contradiction to face convexity. This implies that at some point we reach (intersect)
	the~$s$-$v$-subpath of~$\rho$ by going upwards or to the right, and thus
	we construct a distance-decreasing $s$-$t$-path shorter than $\rho$, a contradiction.
\end{proof}

Observe that the property we proved in Theorem~\ref{th:bounded-dilation} cannot be guaranteed for general greedy drawings, even if they are convex. For instance, Fig.~\ref{fig:dilation:example} depicts a greedy (non-orthogonal) convex drawing of a biconnected planar graph in which every distance-decreasing path from vertex $s$ to vertex $t$ can be made arbitrarily longer than the Euclidean distance between $s$ and $t$. This construction shows the existence of a family of greedy convex drawings with unbounded dilation.

\begin{figure}[tb]
	\centering \includegraphics[page=5]{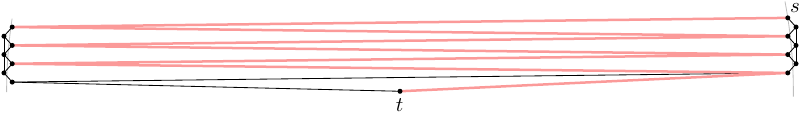}
	\caption{The bold zigzag path is the shortest distance-decreasing
		$s$-$t$-path.}
	\label{fig:dilation:example}
\end{figure}

\subsection{Conflicts in rectilinear representations.} 
\medskip
In this subsection we define the concept of ``conflict'' between two vertices of a rectilinear representation $H$. Intuitively, two vertices are in conflict if they can have different relative positions (either right/left or top/bottom) in different drawings of $H$. Studying the pairs of conflicting vertices of $H$ will be fundamental for our results. To formalize this concept, we first define two directed acyclic graphs (DAGs) $D_x$ and $D_y$ associated with $H$, which have been already used in previous works on orthogonal compaction~\cite{DBLP:journals/comgeo/BridgemanBDLTV00,DBLP:journals/siamcomp/Tamassia87}; refer to Fig.~\ref{fi:dx_dy} for an illustration.

The DAG $D_x$ is obtained from~$H$ by orienting the horizontal edges 
from left to right and by contracting each maximal path of vertical edges 
into a node; $D_y$ is defined symmetrically on the maximal paths of 
horizontal edges. Note that $D_x$ and $D_y$ may have multiple 
edges: also, since the external face of~$H$ is orthoconvex, each of $D_x$ and $D_y$ 
has a single source and a single sink, that is, it is an $st$-digraph. 
For any vertex~$u$ of~$H$, we denote by $c_x(u)$ (by $c_y(u)$) the node 
of~$D_x$ (of $D_y$) corresponding to the maximal vertical (horizontal, respectively) 
path containing $u$ in $H$. If $c_x(u) \neq c_x(v)$, 
the notation $u \prec_x v$ ($u \nprec_x v$) denotes the existence (absence)
of a directed path from $c_x(u)$ to $c_x(v)$ in $D_x$. The notation $u \sim_x v$ 
means that either $u \prec_x v$ or $v \prec_x u$ holds, while $u \not\sim_x v$ 
means that none of them holds. The notations $u \prec_y v$, $u \nprec_y v$, 
$u \sim_y v$, and $u \not\sim_y v$ are symmetric for $D_y$. 
Clearly,~$\prec_x$ and $\prec_y$ are  transitive relations.

Lemma~\ref{le:comparable-x-y} 
shows that there is a directed path between any 
two vertices of $H$ in at least one 
of~$D_x$ and~$D_y$. For this, we first give an auxiliary lemma.

\begin{lemma}\label{le:same-coordinate}
	Let $\Gamma$ be a drawing of a \convex rectilinear representation $H$. For two distinct nodes $c_x(u)$ and $c_x(v)$ of $D_x$ such that there exists a horizontal line crossing both the vertical paths corresponding to $c_x(u)$ and to $c_x(v)$ in $\Gamma$, we have that $u \sim_x v$. A symmetric property holds for the nodes of $D_y$. 
\end{lemma}
\begin{proof}
	We give the proof for the first part of the statement; the argument for the 
	second~part is symmetric. Consider a horizontal line crossing both the 
	vertical paths corresponding to~$c_x(u)$ and $c_x(v)$ in $\Gamma$. Let $s_x$ 
	be the portion of this line between the two vertical paths. If $s_x$ does not 
	traverse any face, then it overlaps with a set of horizontal edges in $\Gamma$. 
	Thus, there is a path in~$D_x$ between~$c_x(u)$ and $c_x(v)$. Otherwise, 
	since every face of $H$ is rectangular, there exists a path in $D_x$ 
	between~$c_x(u)$ and $c_x(v)$ whose internal vertices are the vertical paths 
	containing the vertical edges of the faces traversed by $s_x$.
\end{proof}

\begin{lemma}\label{le:comparable-x-y}
For any two vertices $u$ and $v$ of a \convex rectilinear representation~$H$, at least one of the following holds:
\begin{enumerate*}[label=(\roman*)]
	\item\label{it:comp-u<xv} $u \sim_x v$ or
	\item\label{it:comp-u<yv} $u \sim_y v$.
\end{enumerate*}
\end{lemma}
\begin{proof}
	If $c_x(u) = c_x(v)$, then $u$ and $v$ belong to the same vertical path, 
  and thus $u \sim_y v$. The case $c_y(u) = c_y(v)$ is symmetric. Suppose now 
  that $c_x(u)\neq c_x(v)$ and~$c_y(u)\neq c_y(v)$. Let~$\Gamma$ be any rectilinear drawing 
  of~$H$, and assume that $u$ is below and to the left of $v$; the other cases 
  are symmetric. Consider a maximal path~$\pi=(u_1,\ldots,u_k)$ in $H$, 
  with~$u_1=u$, such that for each edge $(u_i,u_{i+1})$, $u_i$ lies either below or to 
  the left of $u_{i+1}$. Since all internal faces of~$H$ are rectangles, $u_k$ 
  is a top-right corner of~$H$ on the external face.

	Note that, if there is a node $u_i$ in $\pi$ such that $c_x(u_i) = c_x(v)$, then $u \prec_x v$ and $u \prec_y v$. Otherwise, $\pi$ crosses either the horizontal line $\ell_h$ through $v$ (to the left of $v$) or the vertical line through $v$ (below $v$). In the former case, let $u_i$ be a node of $\pi$ such that the vertical path containing $u_i$ is crossed by $\ell_h$. Note that, by the construction of $\pi$, either $c_x(u)=c_x(u_i)$ or $u\prec_x u_i$. Also, by Lemma~\ref{le:same-coordinate}, $u_i \prec_x v$. By transitivity, $u \prec_x v$. In the latter case, a symmetric argument is used to prove that $u \prec_y v$. This concludes the proof of the lemma.
\end{proof}

\begin{figure}[tb]
	\centering
	\subcaptionbox{\label{fi:non-universal}}{\includegraphics[page=1]{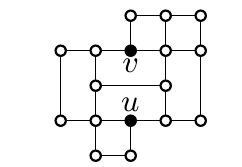}}
	\hfill
	\subcaptionbox{\label{fi:dx}}{\includegraphics[page=2]{dx-dy}}
	\hfill
	\subcaptionbox{\label{fi:dy}}{\includegraphics[page=3]{dx-dy}}
	\hfill
	\subcaptionbox{\label{fi:rigid-non-rectangular}}{\includegraphics[page=4]{dx-dy}}
	\caption{(a) A greedy realizable \convex rectilinear representation $H$ that is not universal greedy due to conflict $\{u,v\}$; (b)--(c) The DAGs $D_x$ and $D_y$ for $H$. (d) A non-\convex representation such that $u \in \cell(v)$ for any drawing, even though $u \prec_x v$ and $u \prec_y v$.}
	\label{fi:dx_dy}
\end{figure}

Let $u$ and $v$ be two vertices of $H$ such that $c_x(u) \neq c_x(v)$ 
and $c_y(u) \neq c_y(v)$. By Lemma~\ref{le:comparable-x-y}, at least one of 
$u \sim_x v$ and $u \sim_y v$ holds. 
If both such conditions hold, then the relative positions (left/right/top/bottom) of 
$u$ and $v$ are fixed (they are the same in any drawing of $H$); in this 
case, we prove that none of the two vertices lies in the cell of the other in 
any drawing of $H$ (refer to Lemma~\ref{lem:conflict-sat-greedy}). Conversely, this is 
not guaranteed when the two vertices are not comparable in one of the two DAGS, 
that is, when either $u \not\sim_x v$ or $u \not\sim_y v$ (refer to Theorem~\ref{th:universal-char}). 
In this case, we say that $u$ and $v$ form a \emph{conflict}, denoted by $\{u,v\}$; in particular, if 
$u \not\sim_x v$, then $\{u,v\}$ is an $x$-conflict, otherwise it is a $y$-conflict.

Let $\{u,v\}$ be an $x$-conflict. By Lemma~\ref{le:comparable-x-y}, we have that either 
$u \prec_y v$ or $v \prec_y u$; suppose that $u \prec_y v$, the other case is symmetric. 
Consider the topmost vertex $u'$ of the vertical path corresponding to 
$c_x(u)$ and the bottommost vertex $v'$ of the vertical path 
corresponding to $c_x(v)$. We say that $u'$ and~$v'$ are \emph{responsible} 
for the conflict~$\{u,v\}$;
e.g., in Fig.~\ref{fi:non-universal}, $\{u,v\}$ is an $x$-conflict with $u'=u$ and $v'=v$.
The next lemma proves that both $u'$ and $v'$ are flat vertices.

\begin{lemma}\label{le:responsible-flat}
Let $u'$ and $v'$ be the responsible vertices for a conflict $\{u,v\}$. Then, both $u'$ and $v'$ are flat vertices.
\end{lemma}
\begin{proof}
Assume, without loss of generality, that the conflict is an $x$-conflict and that $u \prec_y v$. The convexity of $H$ implies that also $u'\prec_y v'$. We show that $u'$ is a north-oriented flat vertex, the argument for showing that $v'$ is a south-oriented flat vertex is symmetric. Since by definition $u'$ is the topmost vertex  of a vertical path, it must be either north-oriented flat or reflex. In the latter case, $u'$ lies on the outer face of $H$; in particular, it lies on the top boundary of the orthoconvex polygon delimiting the outer face of $H$; assume that it also lies on the left side, the other case is symmetric. Since $u' \prec_y v'$, we have that $v'$ lies in the part of the orthoconvex polygon that is above $u'$. Since every point of this part also lies to the right of $u'$, we have that $u' \prec_x v'$, a contradiction. Therefore, $u'$ is a north-oriented flat vertex.
\end{proof}

A conflict is \emph{\sat} in a drawing $\Gamma$ of~$H$ if none of the two 
vertices that are responsible for it lies in the cell of the other. Finally, a convex rectilinear representation
$H$ is \emph{conflict-free} if it has no conflict.

\begin{lemma}\label{lem:conflict-sat-greedy}
  Let $H$ be a \convex rectilinear representation of a biconnected graph. A 
  rectilinear drawing $\Gamma$ of $H$ is greedy if and only if every conflict 
  is \sat~in~$\Gamma$.
\end{lemma}
\begin{proof}
	By Theorem~\ref{th:papa-charact}, drawing $\Gamma$ is greedy if and only if 
  for any vertex $v$ of~$H$, we have that $\cell(v)$ contains no vertex 
  distinct from $v$. This already proves the necessity, since a conflict that 
  is not \sat implies that a vertex lies in the cell of another vertex, by
  definition.
	
	We now prove the sufficiency. First note that, if~$v$ is a vertex on the 
	external face, then the portion of~$\cell(v)$ that belongs to the external 
	face is empty, since the external boundary defines an orthoconvex polygon. 
	Also, since all internal faces of~$H$ are rectangles, there is no internal 
	angle of $270$ degrees. Further, by Property~\ref{pr:right-angles}, if two 
	edges incident to a vertex~$v$ create an angle of $90$ degrees, then the 
	portion of $\cell(v)$ delimited by these two edges is always empty. Thus, the 
	only possible vertices whose cells may be non-empty in $\Gamma$ 
	are the flat vertices.
	Let $v$ be a flat vertex, and assume that the flat angle at~$v$ is south-oriented (the other cases are symmetric). Consider any other vertex~$u$. If~$u$ 
	and~$v$ are not in an $x$-conflict, then either $u\prec_x v$ or~$v\prec_x u$, say 
	the former. Then,~$u$ lies to the left of $v$ in $\Gamma$; also, $u$ does not 
	lie to the right of the left neighbor of $v$, which implies that it lies to 
	the left of $\cell(v)$. Finally, if $u$ and $v$ are in an $x$-conflict, then this is \sat by assumption. Hence, $u \notin \cell(v)$ 
	by definition. Repeating this argument for every flat vertex proves the statement.
\end{proof}

\section{Universal Greedy Rectilinear Representations}\label{se:universal}

We recall that a rectilinear representation $H$ is universal greedy if every rectilinear drawing of $H$ is greedy.
In this section, we first provide a linear-time algorithm to test whether a rectilinear representation is universal greedy (refer to Section~\ref{sse:testing-universal}) and then we describe a full generative scheme for this family of representations (refer to Section~\ref{sse:generative-universal}).

\subsection{Testing algorithm}\label{sse:testing-universal} 
\medskip
Our algorithm to test whether a convex rectilinear representation is universal greedy is based on the following concise characterization. 

\begin{theorem}\label{th:universal-char}
	Let $H$ be a \convex rectilinear representation of a biconnected plane graph. Then, $H$ is universal greedy if and only if it is conflict-free.
\end{theorem}
\begin{proof}
By Lemma~\ref{lem:conflict-sat-greedy}, if $H$ is conflict-free, every rectilinear drawing of $H$ is greedy (note that a rectilinear representation may be conflict-free without being \convex, which would imply that it is not universal greedy; see Fig.~\ref{fi:rigid-non-rectangular}).
	
We now prove the other direction. Suppose, for a contradiction, that $H$ is universal greedy but not conflict-free. Let $\Gamma$ be 
any rectilinear drawing of $H$. Consider two vertices $u$ and $v$ that are 
responsible for a conflict in $H$; assume without loss of generality that $\{u,v\}$
is an $x$-conflict, that is, $u \not\sim_x v$. We can further assume that $u$ and~$v$ are consecutive along the $x$-axis
in $\Gamma$, that is, there is no vertex $w$ such that $x(u) < x(w) < x(v)$. Indeed, if such a vertex $w$ 
exists (which implies $w \nprec_x u$ and $v \nprec_x w$), at least one of $u \nprec_x w$ and $w \nprec_x v$ 
holds, as otherwise $u \prec_x v$. Hence, we could have selected either $u$ 
and $w$ or $w$ and~$v$ instead of $u$ and $v$.

First observe that, if $x(u) = x(v)$, then $\Gamma$ is not greedy, since $u \in \cell(v)$ 
and $v \in \cell(u)$. On the other hand, if $x(u) < x(v)$, then we can transform $\Gamma$ into a 
drawing $\Gamma'$ of $H$ by moving~$u$ and all the vertices in its vertical path to the right
of a quantity $x(v)-x(u)$, so that $x(u)=x(v)$. Since~$u$ and~$v$ are consecutive along the $x$-axis 
in $\Gamma$ and since $H$ is convex, $\Gamma'$ is still planar but not 
greedy, which contradicts the fact that $H$ is universal greedy.
\end{proof}

Before giving our testing algorithm, we observe that it is possible to state an alternative characterization of universal greedy representations as a corollary of Theorem~\ref{th:universal-char}. Namely, suppose that $H$ is a convex rectilinear representation of a biconnected plane graph and suppose that there exists a staircase path from any two vertices $u$ and $v$. This immediately implies that either $u$ and $v$ belong to the same horizontal or vertical path in $H$ (i.e., $c_x(u) = c_x(v) \lor c_y(u) = c_y(v)$), or there exist two directed paths connecting $c_x(u)$ and $c_x(v)$ in $D_x$ and $c_y(u)$ and $c_y(v)$ in $D_y$. Thus, in this case, $H$ is conflict-free. It is not difficult to prove that the reverse is also true, which thanks to Theorem~\ref{th:universal-char} implies the following:   

\begin{corollary}\label{co:staircase-char}
	$H$ is universal greedy if and only if there exists a staircase path between any two vertices of~$H$.
\end{corollary}

\noindent We now present our efficient testing algorithm.

\begin{theorem}\label{th:universal-test}
Let $H$ be a rectilinear representation of an $n$-vertex biconnected plane 
graph. It can be tested in $O(n)$ time whether $H$ is universal 
greedy.
\end{theorem}
\begin{proof}
The algorithm first checks in linear time whether $H$ is \convex. If not, the 
instance is rejected. 
Otherwise, it checks whether both $D_x$ and $D_y$ contain a (directed) 
Hamiltonian path, which can be done in linear time in the size of~$D_x$ and~$D_y$, 
which is $O(n)$. Namely, since each of $D_x$ and $D_y$ is an $st$-digraph, computing a longest path from $s$ to $t$ is done in $O(n)$ time from 
a topological sorting.
We claim that $H$ is universal greedy if and only if this test succeeds. By 
Theorem~\ref{th:universal-char}, to prove this claim, it is enough to show 
that a DAG $D$ contains a Hamiltonian path if and only if for any two 
vertices $u$ and $v$ of $D$, there is a directed path either from $u$ to $v$ 
or from $v$ to~$u$. If~$D$ has a Hamiltonian path~$\pi$, a directed path 
between any two vertices of $D$ is a subpath of $\pi$. 
Conversely, suppose that there exists a directed path between any two vertices 
of~$D$. Then, we can construct a topological sorting of $D$, which determines a 
total order of its nodes, and hence a Hamiltonian path.
\end{proof}

Since conflict-free rectilinear representations form a subclass of the \emph{turn-regular} orthogonal representations~\cite{DBLP:journals/comgeo/BridgemanBDLTV00}, for which a minimum-area drawing can be found in linear time, we can also state the following as a corollary of Theorem~\ref{th:universal-test}.

\begin{corollary}\label{co:universal-min-area}
Let $H$ be a universal greedy rectilinear representation. There is a linear-time algorithm to compute a (greedy) drawing of $H$ with minimum area.
\end{corollary}

\subsection{Generative scheme}\label{sse:generative-universal}
\medskip
We now describe a generative scheme to obtain any possible universal greedy rectilinear representation, starting from a rectangle, by applying a suitable sequence of primitive operations, which incrementally add simple paths on the external face, possibly subdividing external edges.

\begin{figure}[t]
	\centering
	\subcaptionbox{\label{fi:1-reflex}}{\includegraphics[page=1]{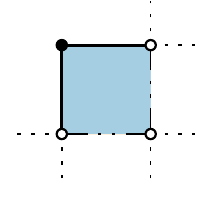}}
	\hfill
	\subcaptionbox{\label{fi:2-reflex-1}}{\includegraphics[page=2]{universal-operations}}
	\hfill
	\subcaptionbox{\label{fi:2-reflex-2}}{\includegraphics[page=3]{universal-operations}}
	\hfill
	\subcaptionbox{\label{fi:3-reflex-1}}{\includegraphics[page=4]{universal-operations}}
	\hfill
	\subcaptionbox{\label{fi:3-reflex-2}}{\includegraphics[page=5]{universal-operations}}
	\hfill
	\subcaptionbox{\label{fi:4-reflex}}{\includegraphics[page=6]{universal-operations}}
	\caption{Schematic illustration of $k$-reflex vertex additions. The new face introduced by the~operation is shaded; the new $k$-reflex vertices are in black. (a) $k=1$; (b)-(c) $k=2$; (d)-(e) $k=3$; (f) $k=4$. }\label{fi:k-reflex-vertex-addition}
\end{figure}

Let $H$ be a biconnected universal greedy rectilinear representation. Each of the following operations on $H$ produces a new biconnected universal greedy rectilinear representation, as Lemma~\ref{le:primitives} proves.

\begin{itemize}
\item[$-$] \emph{$k$-reflex vertex addition.} Attach to the 
external face of $H$ a path of $ 1 \leq k \leq 4$ reflex vertices (corners) 
that forms a new rectangular internal face, provided that the resulting 
representation is \convex (see Fig.~\ref{fi:k-reflex-vertex-addition}).
\item[$-$] \emph{flat vertex addition.} Subdivide an external edge $(u,v)$ of $H$ with a flat vertex of degree two, provided that the open strip of the plane between the two lines orthogonal to $(u,v)$ and passing through $u$ and $v$, respectively, has no vertices in its interior. 
\end{itemize}

Figure~\ref{fi:universal-construction} shows an example of universal greedy representation generated through a sequence of $k$-reflex and flat-vertex additions.

\begin{figure}[tb]
	\centering
	\subcaptionbox{\label{fi:H0}}{\includegraphics[page=1, trim = .8cm .5cm 1.7cm 0, clip]{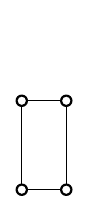}}
	\hfil
	\subcaptionbox{\label{fi:H1}}{\includegraphics[page=2, trim = .8cm .5cm 1.7cm 0, clip]{universal-construction}}
	\hfil
	\subcaptionbox{\label{fi:H2}}{\includegraphics[page=3, trim = .5cm .5cm 1.8cm 0, clip]{universal-construction}}
	\hfil
	\subcaptionbox{\label{fi:H3}}{\includegraphics[page=4, trim = .5cm .5cm 1.8cm 0, clip]{universal-construction}}
  
  \medskip
	\subcaptionbox{\label{fi:H4}}{\includegraphics[page=5, trim = .5cm 0 1cm 0, clip]{universal-construction}}
	\hfil
	\subcaptionbox{\label{fi:H5}}{\includegraphics[page=6, trim = .5cm 0 .5cm 0, clip]{universal-construction}}
	\hfil
	\subcaptionbox{\label{fi:H6}}{\includegraphics[page=7, trim = .5cm 0 .5cm 0, clip]{universal-construction}}
	\hfil
	\subcaptionbox{\label{fi:H7}}{\includegraphics[page=8]{universal-construction}}
	\caption{A sequence of primitive operations that generates a universal greedy rectilinear representation. 
		(\subref{fi:H0})~A single rectangular face; 
		(\subref{fi:H1})~flat vertex addition; 
		(\subref{fi:H2})~$2$-reflex vertex addition; 
		(\subref{fi:H3})~$1$-reflex vertex addition; 
		(\subref{fi:H4})-(\subref{fi:H5})~$3$-reflex vertex additions; 
		(\subref{fi:H6})~$2$-reflex vertex addition; 
		(\subref{fi:H7})~$4$-reflex vertex addition.}
	\label{fi:universal-construction}
\end{figure}

\begin{lemma}\label{le:primitives}
	Let $H$ be a universal greedy rectilinear representation of a biconnected 
	plane graph. Let $H'$ be the rectilinear representation obtained from $H$ by 
	applying either a $k$-reflex vertex addition or a flat vertex addition. Then, $H'$ 
	is biconnected and it is a universal greedy rectilinear representation. 
\end{lemma}
\begin{proof} %
	To prove that $H'$ is biconnected, note that subdividing an edge or attaching a simple path 
	between two vertices of a biconnected graph cannot create cutvertices. We claim that $H'$ is also \convex. In fact, a flat vertex addition does not change the shape of any face of~$H$, while a $k$-reflex vertex addition creates a new rectangular face and maintains the property that the the outer face is delimited by an orthoconvex polygon.
	
	We now show that $H'$ is universal greedy. By Theorem~\ref{th:universal-char}, this is equivalent to proving that $H'$ is conflict-free. For this, we show that none of the two operations introduces conflicts to $H$, which is universal greedy and thus conflict-free.
	
	Consider first a flat vertex addition that subdivides an edge $(u,v)$ into two edges $(u,w)$ 
	and $(w,v)$. Without loss of generality, assume that $(u,v)$ is horizontal, with 
	$u$ to the left of $v$. Since $c_y(w) = c_y(u) = c_y(v)$, vertex $w$ cannot form a $y$-conflict with any other vertex in $H'$. Also, by hypothesis, the open strip between the two lines orthogonal to $(u,v)$ and passing through $u$ and $v$, respectively, does not contain vertices in its interior. Hence, there cannot be vertices in $H'$ that form an $x$-conflict with $w$. It follows that $H'$ is still conflict-free.   
	
	Assume now that~$H'$ is obtained by applying a $k$-reflex vertex addition to $H$. Let $u$ and $v$ be the two vertices of $H$ that are joined by a path $u,w_1,\dots,w_k,v$, with $1 \leq k \leq 4$, in order to obtain $H'$. Suppose, for a contradiction that there exists a conflict in $H'$, and let $a$ and $b$ be the vertices responsible for it. First observe that $a,b \notin \{w_1,\dots,w_k\}$, since $w_i$ is a reflex vertex, for each $i=1,\dots,k$, while $a$ and $b$ are flat vertices by Lemma~\ref{le:responsible-flat}. 
	
	We now consider the case in which one of $a$ and $b$ coincides with one of $u$ and $v$, say $a = u$, and the flat angle at $u$ involved in the conflict $\{a,b\}$, call it~$\phi$, is delimited by edge $(u,w_1)$; see Figs.~\ref{fi:2-reflex-2}--\ref{fi:4-reflex}. 
	We claim that, in any drawing $\Gamma'$ of $H'$, the part of the cell of $u$ that is determined by $\phi$ does not contain $b$, which implies that $u$ is not in conflict with $b$, a contradiction. The claim follows from the fact that this part of the cell of $u$ is a 
    subset of the cell of $u$ in the drawing of $H$ obtained by removing $w_1,\dots,w_k$ from $\Gamma'$, which is empty as $H$ is universal, and from the fact that $b \in H$. Symmetrically, we can prove that $v$ is not responsible for any conflict due to a flat angle delimited by edge $(v,w_k)$. Hence, the conflict $\{a,b\}$ is determined by two flat angles that also exist in $H$, contradicting the fact that $H$ is conflict-free. Therefore, we have that $H'$ is conflict-free, and thus universal by Theorem~\ref{th:universal-char}.  
\end{proof}

\noindent The next lemma is used to prove Theorem~\ref{th:generative-scheme}.

\begin{lemma}\label{le:biconnected}
	Let $H$ be a rectilinear representation of a biconnected plane graph~$G$, 
	such that all internal faces of $H$ are rectangles.
	If $G$ is not a simple cycle, then there exists an internal face $f$ of $G$ 
	such that: 
	\begin{enumerate}[label=(\roman*)]
		\item \label{prp:1} $f$ is adjacent to the external face of $G$; 
		\item \label{prp:2} $f$ has a degree-$2$ vertex that is a reflex vertex in the 
		external face of~$H$;
		\item \label{prp:3} $G$ remains biconnected if we remove from it all the external degree-$2$ vertices of $f$ and their 
		incident edges. 
	\end{enumerate}
\end{lemma}
\begin{proof}
	Let $G^*$ be the \emph{weak dual} of $G$, i.e., the node set of $G^*$ is the 
	set of the internal faces of $G$, and for each edge $e$ of $G$ shared by two 
	internal faces $f$ and~$g$, there is a dual edge of $e$ in $G^*$ that 
	connects the two nodes corresponding to~$f$ and $g$. Note that $G^*$ can be obtained from the (non-weak) dual of $G$ by deleting the node corresponding to the external face of $G$; hence, $G^*$ is (at least) connected, since the dual of $G$ is biconnected (see, e.g.,~\cite{tutte-66}). 
	Also, as soon as $G$ becomes non-biconnected due to the removal of some edges of the 
	external face, then $G^*$ becomes disconnected. Indeed, in this case, $G$ 
	would have a cutvertex $c$ on the external face, which means that in $G^*$ 
	there would be no path between any two nodes corresponding to faces that belong 
	to different biconnected components of $G$ with respect to $c$. Therefore, it 
	is sufficient to prove that there exists a face $f$ in $G$ that verifies 
	properties $(i)$ and $(ii)$, and such that $G^*$ remains connected after the 
	removal from $G$ of all the external degree-$2$ vertices of $f$. To this aim, 
	we distinguish between two cases, based on whether $G^*$ is biconnected or 
	simply connected.
	
	\newcase
	\ccase{c:biconnected} $G^*$ is biconnected. Since the external boundary of $H$ is a rectilinear polygon, the external face of $H$ has at least four reflex vertices. Let $v$ be one of them and~$f$ be the internal face containing $v$.
	Removing from $G$ all the external degree-$2$ vertices of $f$ (included $v$) causes the removal of the node corresponding to $f$ in $G^*$. Since $G^*$ was biconnected, it remains connected after such a removal.
	
	\ccase{c:connected} $G^*$ is connected but not biconnected. Let $\mathcal{T}$ 
  be the block-cutvertex tree of $G^*$, and let~$B^*$ be a block of $G^*$ that 
  is a leaf B-node of $\mathcal{T}$. Hence,~$B^*$ contains only one cutvertex of $G^*$, 
  which corresponds to an internal face $f_c$ of $G$. Denote by $F_c$ the set 
  of internal faces of $G$ distinct from $f_c$ and whose corresponding nodes of 
  $G^*$ are in~$B^*$. It can be seen that there is a face $f \in F_c$ that 
  contains a reflex vertex in the external face of $H$. More precisely, let $s$ 
  be the number of sides of $f_c$ in $H$ that are incident to some face of $F_c$,
  and let $r$ be the number of reflex vertices in the external face of $H$ 
  that belong to some faces of $F_c$. Since the boundary of the rectilinear 
  representation $H$ restricted to $F_c$ is a rectilinear polygon, we have that: 
	\begin{enumerate*}[label=(\alph*)]
		\item if $s = 1$, then $r \geq 2$; 
		\item if $s = 2$, then $r \geq 3$; 
		\item if $s \in \{3,4\}$, then $r \geq 4$.
	\end{enumerate*} 
	Hence, removing from $G$ all the external degree-$2$ vertices of $f$ causes 
  the removal of the node corresponding to $f$ in $B^*$. Since $B^*$ was 
  biconnected, it remains connected after such a removal, and $G^*$ remains 
  connected as well.
\end{proof}

\noindent We are now ready to prove Theorem~\ref{th:generative-scheme}.

\begin{theorem}\label{th:generative-scheme}
Let $H$ be a universal greedy rectilinear representation of a biconnected planar graph. Then, $H$ can be obtained by a suitable sequence of $k$-reflex vertex and flat vertex additions, starting from a rectangle. 
\end{theorem}
\begin{proof}
	We prove that there exists a sequence $H_0, H_1, \dots, H_r$ $(r \in \mathbb{N})$ of universal greedy rectilinear representations such that $H_0$ is a rectangle, $H_r = H$, and $H_{i+1}$ is obtained by applying either a $k$-reflex vertex addition or a flat vertex addition  on $H_i$ $(i = 0, \dots, r-1)$. To this aim, it suffices to show that from each $H_{i+1}$ we can derive $H_i$ by applying a reverse operation of either a $k$-reflex vertex addition or a flat vertex addition. We distinguish between two cases:
	
	\newcase
	\ccase{c:flat} $H_{i+1}$ has a flat degree-$2$ vertex $w$ on the external 
  face. Let $u$ and $v$ be the neighbors of $w$. Let $H_i$ be the rectilinear 
  representation obtained from~$H_{i+1}$ by deleting the edges $(u,w)$ and $(w,v)$, 
  and by adding the edge $(u,v)$ (as a single segment). Clearly, $H_i$ 
  remains biconnected, \convex, and greedy universal.  
	Also, $H_{i+1}$ is obtained from $H_i$ by applying a flat vertex addition 
  that subdivides~$(u,v)$. 
	
	\ccase{c:reflex} Every degree-$2$ vertex on the external face of $H_{i+1}$ is 
	a reflex vertex.
	Note that the external face contains at least four reflex vertices. Let $f$ 
	be an internal face having the Properties \ref{prp:1}--\ref{prp:3} in the statement of 
	Lemma~\ref{le:biconnected} (this lemma guarantees that such a face exists). 
	By the proof of Lemma~\ref{le:biconnected}, the external degree-$2$ vertices 
	of $f$ form a path~$\pi$, and their removal preserves biconnectivity. Since 
	by hypothesis there is no external flat vertex of degree two in $H_{i+1}$, 
	all vertices of $\pi$ are reflex vertices in the external face of~$H_{i+1}$. 
	Also, since~$f$ is rectangular, $\pi$ is formed by at most $k$ vertices, with 
	$k \in \{1,2,3,4\}$. Now, let~$u$ and $v$ be the two vertices of $f$ to which 
	$\pi$ is attached, and let $\pi'$ be the path from~$u$ to $v$ containing all 
	the internal edges of~$f$ (the boundary of~$f$ is the union of $\pi$ and $\pi'$). 
	Since~$H_{i+1}$ is universal greedy, $\pi'$ cannot contain two vertices 
	with an angle of $90$ degrees inside $f$ (i.e.,~$\pi'$ is either a straight-
	line path or it is an $L$-shaped path). Indeed, in such a case, $u$ and~$v$ 
	would be two flat vertices on opposite sides of $f$, which, as already 
	observed, contradicts the fact that~$H_{i+1}$ is universal greedy.
	Let $H_i$ be the rectilinear representation derived from~$H_{i+1}$ by 
	removing~$\pi$. For the above properties, $H_i$ remains \convex. Also,~$H_{i+1}$ 
	is universal greedy, because $\pi'$ is a staircase path from $u$ to $v$ 
	and thus every staircase path of~$H_{i+1}$ that contains $\pi$ can be 
	replaced with a staircase path in which~$\pi$ is substituted with~$\pi'$. 
	This proves that~$H_{i+1}$ is obtained from~$H_i$ by applying a $k$-reflex 
	vertex addition.
\end{proof}

\section{General Greedy Rectilinear Representations}\label{se:greedy-rectilinear}

In this section, we consider \convex rectilinear representations of 
biconnected plane graphs that may contain conflicts. In particular, we investigate conditions under which a biconnected plane graph $H$ is greedy realizable. We present
a characterization (refer to Theorem~\ref{thm:characterization}), which yields a 
polynomial-time testing algorithm for a meaningful subclass of instances, namely 
when $D_x$ and $D_y$ are series-parallel (refer to Theorem~\ref{thm:seriesparallel}). 

Let $D$ be one of the two DAGs $D_x$ and $D_y$ associated with $H$. Since $D$ is an $st$-digraph, it has an $st$-ordering $\mathcal{S} = v_1, \dots, v_m$. For two indices $i$ and $j$, with $1 \leq i < j \leq m$, $D\langle i,j\rangle$ denotes the subgraph of $D$ induced by $v_i, \dots, v_j$. We say that $\mathcal{S}$ is \emph{\good} if: 

\begin{enumerate}[label=S.\arabic*]
  \item\label{c:strips-2comp} For any two indices $i$ and $j$, with $1 \leq i < j \leq m$, $D\langle i,j\rangle$ has at most two connected components, and
  \item\label{c:strips-precede} if $D\langle i,j\rangle$ has exactly two components, then all nodes of one component precede those of the other in~$\mathcal{S}$.
\end{enumerate}

Further, we say that a drawing of $H$ \emph{respects} an $st$-ordering $\mathcal{S}_x$ of $D_x$ ($\mathcal{S}_y$ of~$D_y$) if for any two vertices $u$ and $w$ of $H$, we have that $u$ lies to the left of~$w$ (below~$w$) in the drawing if and only if $c_x(u)$ precedes $c_x(w)$ in $\mathcal{S}_x$ ($c_y(u)$ precedes~$c_y(w)$ in $\mathcal{S}_y$).
Finally, when we refer to the $x$-coordinate ($y$-coordinate) of a node $v_i$ of $D_x$ (of $D_y$), we mean the one of all the vertices $w \in H$ with $c_x(w) = v_i$ (with $c_y(w) = v_i$), as these vertices belong to the same vertical (horizontal) path.
We prove the following characterization. 

\begin{theorem}\label{thm:characterization}
A \convex rectilinear representation $H$ of a biconnected plane graph is greedy realizable if and only if both DAGs $D_x$ and $D_y$ admit \good $st$-orderings.
\end{theorem} 

The following three subsections are devoted to the proof of Theorem~\ref{thm:characterization}. In particular, we prove the necessity of the existence of \good $st$-orderings in Section~\ref{sse:necessity}; then, for a proof of the sufficiency, we first discuss in Section~\ref{sse:conflicts} some properties of greedy rectilinear representations with respect to their conflicts, and then we use these properties in Section~\ref{sse:algorithm} to derive a drawing algorithm, assuming a \good $st$-ordering.

\subsection{Necessity of the condition in Theorem~\ref{thm:characterization}} \label{sse:necessity}
\medskip
In this section, we prove that the existence of \good $st$-orderings for both DAGs is a necessary condition for $H$ to be greedy realizable.

\begin{lemma}\label{le:necessity}
If $D_x$ or $D_y$ admits no \good $st$-ordering, $H$ is not greedy realizable.
\end{lemma}
\begin{proof}
	Let $\mathcal{S}_x$ be any $st$-ordering of $D_x$ that is not \good. We prove 
	that $H$ does not admit any greedy drawing respecting $\mathcal{S}_x$. 
	Suppose, for a contradiction, that there exists such a greedy drawing $\Gamma$
	of $H$. Since~$\mathcal{S}_x$ is not \good, there exist two indices $i$ 
	and $j$, with $1 \leq i < j \leq m$, such that $D_x\langle i,j\rangle$ consists 
	of at least two connected components.
	
	Let $\ell_A$ be a vertical line with $x$-coordinate between~$x(v_{i-1})$ and~$x(v_{i})$,
	and let~$\ell_B$ be a vertical line with $x$-coordinate between~$x(v_{j})$
	and~$x(v_{j+1})$ in $\Gamma$.
	Observe that, for a connected component $C$ of $D_x\langle i,j\rangle$, the 
	following property holds. Consider the smallest rectangle $R(C)$ having its 
	vertical sides along $\ell_A$ and~$\ell_B$ and containing all the vertices 
	of $H$ corresponding to nodes of~$C$ in its interior; then, every horizontal 
	segment connecting two points on the two vertical sides of $R(C)$ intersects 
	at least a vertical edge between two vertices $u$ and $w$ of $H$ such 
	that $c_x(u) = c_x(w) \in C$. In fact, if this was not the case, then there would be two vertices of $C$ that are not joined by any path in $C$, contradicting the fact that $C$ is a connected component.
	
	Let $C_1$ and $C_2$ be two components of $D_x\langle i,j\rangle$, and consider 
  two vertices~$u_1$ and $u_2$ of~$H$ such that $c_x(u_1) \in C_1$  
  and $c_x(u_2) \in C_2$; see Fig.~\ref{fig:notrealizable-2comp-app}. This 
  implies that $u_1 \nsim_x u_2$. Thus, by Lemma~\ref{le:comparable-x-y}, 
	either $u_1 \prec_y u_2$ or $u_2 \prec_y u_1$ holds; assume the latter. Consider 
	another pair of vertices $u_1'$ and $u_2'$ of $H$ such that $c_x(u_1') \in C_1$ 
	and $c_x(u_2')\in C_2$. By the same argument, either $u_1' \prec_y u_2'$ 
  or $u_2' \prec_y u_1'$ holds; we claim that $u_2' \prec_y u_1'$. Suppose for a 
	contradiction that $u_1' \prec_y u_2'$. 
	Consider the two rectangles $R(C_1)$ and $R(C_2)$ as defined above. Note 
	that, since $u_2 \prec_y u_1$ and $u_1' \prec_y u_2'$, we 
	have $R(C_1) \cap R(C_2) \neq \emptyset$, as the rectangle $R(C_1)$ 
	contains $u_1$ and~$u_1'$, and thus it contains also $u_2$ and~$u_2'$. 
	Therefore, there exist a vertical path corresponding to a node of~$C_1$ and a 
	vertical path corresponding to a node of $C_2$ that are crossed by the same 
	horizontal line. By Lemma~\ref{le:same-coordinate}, there exists a directed path in $D_x$ between 
	the two nodes of $C_1$ and~$C_2$, contradicting the fact that~$C_1$ and~$C_2$ 
	are different connected components.
	Repeating this argument for any pair of vertices, we conclude that there 
	exists a horizontal line-segment~$h$ from $\ell_A$ to $\ell_B$ such that all 
	the vertices of $H$ corresponding to nodes of $C_1$ lie above $h$ and all 
  those corresponding to nodes of $C_2$ lie below $h$ in $\Gamma$.
	
	Further, since for each node of $C_1$ there is a flat vertex of $H$ that is south-oriented (the bottommost vertex of the vertical path corresponding to the node of $C_1$), we have that the union of the cells of these flat vertices, restricted to the region below $h$, consists of a rectangle of infinite height spanning at least all the $x$-coordinates between those of the leftmost and of the rightmost node of $C_1$ (see the tiled region in Fig.~\ref{fig:notrealizable-2comp-app}).
	Since the same holds for the cells of the north-oriented flat vertices that are the topmost points of the vertical paths corresponding to nodes of $C_2$, we have that all the nodes of $C_1$ are to the left of all the nodes of $C_2$ in $\Gamma$, or vice versa.
	Therefore, $D_x\langle i,j\rangle$ contains at least another connected component~$C_3$, as 
	otherwise the $st$-ordering $\mathcal{S}_x$ would be \good. By the 
	same argument as before, we can claim that $C_1$, $C_2$, and $C_3$ are 
	separated by horizontal line-segments; we further assume that $C_1$, $C_2$, and $C_3$ appear 
	in this order from top to bottom in $\Gamma$. Also, either all the nodes
	of $C_3$ lie to the left of all the nodes of $C_2$ in $\Gamma$, or vice 
	versa, and the same holds for the nodes of $C_3$ and of~$C_1$; 
	see Fig.~\ref{fig:notrealizable-3comp-app}.
	
	\begin{figure}[t]
		\centering
		\subcaptionbox{\label{fig:notrealizable-2comp-app}}
		{\includegraphics[page=1]{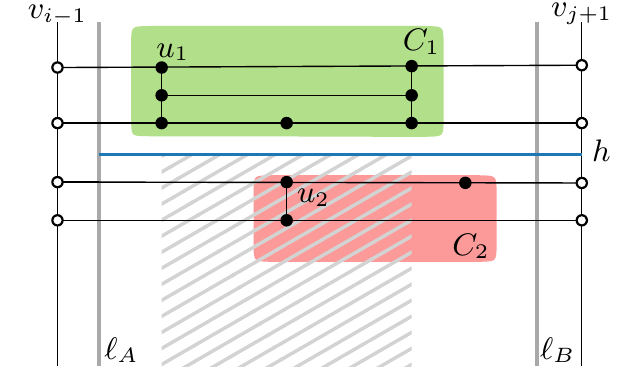}}
		\hfil
		\subcaptionbox{\label{fig:notrealizable-3comp-app}}
		{\includegraphics[page=2]{notrealizable}}
		\caption{Illustration for the proof of Lemma~\ref{le:necessity}.
			(a)~There are two components in~$D_x\langle i,j\rangle$. $C_1$ and~$C_2$ are separated by~$h$. The cells of the south-oriented flat vertices of~$C_1$ span at least the tiled area. (b)~There are at least three components. There are vertices of~$C_2$ to the left of~$\ell_C$. The cell of~$v_r$ spans at least the tiled area.}
		\label{fig:notrealizable-app}
	\end{figure}
	
	Assume that all the nodes of $C_1$ are to the left of all those of $C_2$, 
	which are to the left of those of $C_3$; the other cases are analogous. 
	Let $\ell_C$ be the vertical line that is equidistant from~$\ell_A$ and $\ell_B$. 
	We claim that all the nodes of $C_2$ are required to lie to the right 
	of $\ell_C$. Namely, if there is at least a node of $C_1$ to the right 
	of $\ell_C$, this is trivially true since the nodes of $C_2$ are to the right 
	of those of $C_1$, by assumption. Further, if all the nodes of $C_1$ lie to the 
	left of $\ell_C$, let $v_r$ be the bottommost vertex of the vertical path 
	corresponding to the rightmost node of $C_1$. Let $x_r$, $x_A$, $x_B$, and $x_C$ 
	be the $x$-coordinates of~$v_r$, $\ell_A$, $\ell_B$, and $\ell_C$, respectively. 
	Assuming all positive $x$-coordinates, we have that 
	$x_C - x_r < x_C - x_A = x_B - x_C$; thus, $x_r + (x_B - x_r)/2 > x_C$. This implies 
	that the right boundary of $\cell(v_r)$ lies to the right of $\ell_C$, since 
	the neighbor of $v_r$ in $H$ with its same $y$-coordinate and with larger $x$-coordinate lies to the right of $\ell_B$, as otherwise $c_x(v_r)$ would not be the rightmost node of $C_1$. Hence, the claim follows, since the nodes of $C_2$ must lie to the right of $\cell(v_r)$.
	With analogous arguments we can prove that the topmost vertex of the vertical path corresponding to the leftmost node of $C_3$ enforces all the nodes of $C_2$ to lie to the left of $\ell_C$. 
	This results in a contradiction and concludes the proof.
\end{proof}

We now prove that the necessary condition of Theorem~\ref{thm:characterization} 
is also sufficient for the existence of a greedy rectilinear drawing. Our 
proof is constructive, as we provide a polynomial-time algorithm that, given 
a \good $st$-ordering for each of the two DAGs of a rectilinear 
representation $H$, constructs a greedy drawing~$\Gamma$ of~$H$. Our 
algorithm is based on some properties concerning the conflicts of~$H$, which 
we discuss in the following subsection.

\subsection{Properties of conflicts in greedy rectilinear representations} \label{sse:conflicts}
\medskip
We start by proving a lemma that allows us to assign the $x$- and $y$-coordinates of the vertices in $\Gamma$ in two independent steps.

\begin{lemma}\label{lem:independent-x-y}
Let $H$ be an \convex rectilinear representation of a biconnected embedded planar graph $G$. 
Let $\Gamma_1$ and $\Gamma_2$ be two drawings of $H$
such that all $x$-conflicts are resolved in~$\Gamma_1$ and
all $y$-conflicts are resolved in~$\Gamma_2$. 
Then, the drawing $\Gamma_3$ of $H$ in which the $x$-coordinate of each vertex 
is the same as in $\Gamma_1$ and the $y$-coordinate of each vertex is the same 
as in $\Gamma_2$ is greedy.
\end{lemma}
\begin{proof}
	By Theorem~\ref{th:papa-charact}, in order to prove that $\Gamma_3$ is 
	greedy, it is enough to prove that for any vertex $v$ of $H$, we have 
	that $\cell(v)$ contains no vertex distinct from $v$. Since $H$ is \convex, there 
	is no internal angle of $270$ degrees. Also, by Property~\ref{pr:right-angles}, 
	if two edges incident to a vertex $v$ create an angle of $90$ degrees, then 
	the portion of $\cell(v)$ delimited by these two edges is always empty. Thus, 
	if, for a vertex $v$, the cell $\cell(v)$ is non-empty in~$\Gamma_3$, then
	$v$ forms one or two flat angles. However, the fact that a 
	vertex lies inside a cell determined by a north-oriented or by a south-oriented 
	flat angle only depends on the $x$-coordinates of the vertices in 
	the drawing; thus, all these cells are empty in $\Gamma_3$, since they are 
	empty in $\Gamma_1$. Analogously, all the cells determined by east-oriented 
	or by west-oriented flat angles are empty in $\Gamma_3$, since they are empty 
	in~$\Gamma_2$. This concludes the proof of the lemma.
\end{proof}

In view of Lemma~\ref{lem:independent-x-y}, we only focus on the assignment of the $x$-coordinates based on the \good $st$-ordering $\mathcal S_x=v_1,\ldots,v_m$ of $D_x$, which implies that the only conflicts that we have to consider are the $x$-conflicts. The assignment of the $y$-coordinates based on the \good $st$-ordering of $D_y$ works symmetrically.

We now give an overview of our strategy. We first prove in Lemma~\ref{lem:dominations} that, to guarantee that every $x$-conflict is \sat, it suffices to resolve a specific subset of them, called \emph{minimal}. Namely, we say that an $x$-conflict~$\{u,v\}$ \emph{dominates} an $x$-conflict~$\{w,z\}$, with~$c_x(u)=v_i$, $c_x(v)=v_j$, $c_x(w)=v_k$, and $c_x(z)=v_\ell$, if~$k\le i<j\le\ell$. 
A \emph{minimal} $x$-conflict is not dominated by any $x$-conflict. In Fig.~\ref{fig:dominations:1-app}, the $x$-conflict $\{z,r\}$ is minimal and dominates the $x$-conflict $\{u,w\}$. 

By Lemmas~\ref{lem:conflict-sat-greedy} and~\ref{lem:dominations}, we conclude 
that a greedy rectilinear drawing can be obtained by resolving all the minimal 
conflicts. 
In our algorithm, described in Section~\ref{sse:algorithm}, we encode that a minimal $x$-conflict is \sat with a single 
inequality on the horizontal distances between the vertices in the $x$-conflict. 
Then, in Lemma~\ref{lem:consecutive}, we prove that, for a minimal $x$-conflict $\{u,v\}$,
the nodes~$c_x(u)$ and $c_x(v)$ of $D_x$ are consecutive in~$\mathcal S_x$. 
We use this property to show that the system of inequalities describing 
the conditions for the minimal $x$-conflicts to be \sat always admits a solution. 

\begin{lemma}\label{lem:dominations}
Let $\Gamma$ be a rectilinear drawing of $H$ respecting $\mathcal{S}_x$. If 
every minimal $x$-conflict dominating an $x$-conflict $\{u,w\}$ is \sat in $\Gamma$, 
$\{u,w\}$ is \sat.
\end{lemma}
\begin{proof}
	We may assume without loss of generality that $u$ and $w$ are responsible for 
	$\{u,w\}$. Let $v_i = c_x(u)$ and $v_j = c_x(w)$, with $i < j$. Consider the 
	graph~$D_x\langle i,j \rangle$. Since $\mathcal S_x$ is \good, this graph has 
	at most two connected components $C_1$ and $C_2$. Assume that $v_i \in C_1$.
	
	Suppose first that also $v_j  \in C_1$. Consider the right neighbor $u'$ of $u$ 
	in~$H$, which exists since $u$ is a flat vertex; see Fig.~\ref{fig:dominations:claim-app}. 
	Note that node~$c_x(u')$ precedes~$c_x(w)$ in $\mathcal{S}_x$, 
	that is, $c_x(u') \in D_x\langle i,j \rangle$; in fact, if 
	this were not the case, then~$v_j$ would not belong to $C_1$. Thus, $u'$ lies 
	to the left of $w$ in any rectilinear drawing of $H$ respecting $\mathcal{S}_x$. 
	Hence, the mid-point of edge $(u,u')$, which defines the right boundary of 
	$\cell(u)$, lies to the left of $w$, which implies that $w \not\in \cell(u)$. 
	Symmetrically, we can show that $u \not\in \cell(w)$.
	
	\begin{figure}[tb]
		\subcaptionbox{\label{fig:dominations:claim-app}}{\includegraphics[page=1]{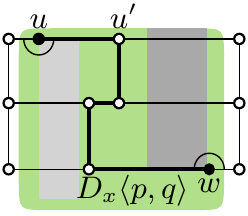}}
		\hfill
		\subcaptionbox{\label{fig:dominations:1-app}}{\includegraphics[page=2]{dominations}}
		\hfill
		\subcaptionbox{\label{fig:dominations:2-app}}{\includegraphics[page=3]{dominations}}
		\hfill
		\subcaptionbox{\label{fig:dominations:3-app}}{\includegraphics[page=4]{dominations}}
		\caption{Illustration for the proof of Lemma~\ref{lem:dominations}.
			(a)~$w\notin\cell(u)$ and $u\notin\cell(w)$;
			(b)~the $x$-conflict~$\{z,r\}$ is \sat;
			(c)~$c_x(u)$ and~$c_x(z)$ are sinks in~$D_x\langle i,j\rangle$; and
			(d)~$c_x(u)$ is not a sink in~$C_1$.}
		\label{fig:dominations-app}
	\end{figure}
	
	Suppose now that $v_j  \in C_2$. As in the proof of Lemma~\ref{le:necessity}, 
	we can assume that all the vertices corresponding to nodes of $C_1$ lie above 
	those corresponding to nodes of $C_2$. Further, we can assume that all the 
	nodes of~$C_2$ follow all those of $C_1$ in the \good $st$-ordering $\mathcal S_x$. 
	The other cases are symmetric.
	Let~$z$ be the bottommost vertex of the vertical path corresponding to the 
	last node~$c_x(z)$ of~$C_1$ in $\mathcal S_x$; see Fig.~\ref{fig:dominations:1-app}.
	Also, let $r$ be the topmost vertex of the vertical path corresponding to the 
	first node $c_x(r)$ of $C_2$ in $\mathcal S_x$. Note that vertices $z$ and $r$ 
	are responsible for a minimal $x$-conflict $\{z,r\}$, which is \sat by assumption. 
	We now show that also $\{u,w\}$ is \sat. In particular, we show 
	that $w \not \in \cell(u)$; the argument for $u \not \in \cell(w)$ is symmetric. 
	
	First observe that, if the right neighbor $u'$ of~$u$ in $H$ belongs to~$C_1$, 
	$\cell(u)$ does not extend beyond~$c_x(z)$; since every node of $C_2$ 
	is completely to the right of $c_x(z)$, we have $w \not\in \cell(u)$; 
	see Fig.~\ref{fig:dominations:1-app}. 
	Thus, we assume that $u$ lies on the right boundary of $C_1$, i.e., its 
	right neighbor $u'$ does not belong to~$C_1$. Note that, if $c_x(u)$ is also 
	a sink of~$C_1$, then~$C_1$ does not contain any other node other than~$c_x(u)$, 
	since $c_x(u)$ is the first node of~$C_1$; see Fig.~\ref{fig:dominations:2-app}. 
	Thus, either~$v_i$ is not a sink of $C_1$, or $c_x(u)=c_x(z)$.
	In the latter case, $r \notin \cell(u)$, since the minimal $x$-conflict $\{z,r\}$ 
	is \sat, which implies $w \notin \cell(u)$. Hence, it remains to consider the 
	case that~$c_x(u)\neq c_x(z)$ and~$c_x(u)$ is not a sink of $C_1$; see 
	Fig.~\ref{fig:dominations:3-app}. This implies that there is a directed path 
	from $c_x(u)$ to $c_x(z)$ in~$C_1$.
	
	Since $u' \not\in C_1$ and since $u$ is a south-oriented flat vertex, $u$ lies below $z$. Consider the right neighbor $z'$ of $z$, which lies to the right of $w$
	because $z$ is the sink of $C_1$. Recall that also $u'$ lies to the right of~$w$.
	Assume first that~$z'$ lies to the left of~$u'$, and let $c_x(z')=v_k$. 
	Consider now the graph~$D_x\langle i,k \rangle$ from $v_i$ to~$v_k$. This 
	graph contains two connected components, one containing $c_x(u)=v_i$ 
	and $c_x(z')=v_k$, and another one containing~$c_x(w) =v_j$, due to the presence of 
	the edge~$(u,u')$, which cannot be crossed. However, this implies a 
	contradiction to Condition~\ref{c:strips-precede} of a \good $st$-ordering, 
	since~$i<j<k$. 
	Thus,~$z'$ must lie to the right of $u'$; since $z$ is to the right of $u$, 
	the right boundary of $\cell(z)$ is to the right of the right boundary of $
	\cell(u)$. Hence, the fact that $r \not \in \cell(z)$ implies 
	that $r \not \in \cell(u)$, and thus $w \not \in \cell(u)$.
\end{proof}

\begin{lemma}\label{lem:consecutive}
	For any two vertices $u$ and $w$ of $H$ such that $\{u,w\}$ is a minimal 
  $x$-conflict, we have that $c_x(u)$ and $c_x(w)$ are consecutive in a \good $st$-ordering~$\mathcal{S}_x$.
\end{lemma} 
\begin{proof}
Suppose that there is a vertex $z \in H$ such that $c_x(z) = v_j$ 
lies between~$c_x(u) = v_i$ and $c_x(w) = v_k$ in $\mathcal{S}_x$, i.e., $i < j < k$. 
First, suppose that~$c_x(u)$ and $c_x(w)$ belong to the same connected component $C$ 
of $D_x\langle i,k\rangle$. Then, by definition, $v_i$ and $v_k$ are a source and 
a sink of $C$, respectively. Since $\{u, w\}$ is an $x$-conflict, we 
have $u \nprec_x w$; hence, there is another source $c_x(s)$ in~$C$, for some 
vertex $s \in H$, such that $s \prec_x w$. Since $c_x(u)$ and $c_x(s)$ are 
different sources of~$C$, we have $u \not\sim_x s$, and thus $\{u,s\}$ is an 
$x$-conflict dominating the minimal $x$-conflict $\{u,w\}$; a contradiction. 
Suppose now that $c_x(u)$ and $c_x(w)$ belong to different components. 
Then,~$c_x(z)$ does not 
belong to the same component as one of them, say $c_x(u)$. 
Thus, $u \not\sim_x z$, i.e., $\{u,z\}$ is an $x$-conflict dominating~$\{u,w\}$; a contradiction. 
\end{proof}

We are now ready to present our algorithm to assign $x$-coordinates to the vertices of $H$ so that all minimal $x$-conflicts are \sat. 

\subsection{A greedy drawing algorithm when $D_x$ and $D_y$ admit \good $st$-orderings} \label{sse:algorithm}
\medskip
We extend some definitions from vertices of~$H$ to nodes of~$D_x$. Namely, we 
say~$v_i\prec_x v_j$, if there is a directed path in~$D_x$ from~$v_i$ and~$v_j$. 
Also, we say that there is a (minimal) $x$-conflict~$\{v_i,v_j\}$ in~$D_x$, if 
there is a (minimal) $x$-conflict~$\{u,w\}$ in~$H$ such that~$c_x(u)=v_i$ and~$c_x(w)=v_j$. 

For $0<i,j\le m$, let $x_{i,j}:=x_{j}-x_i$ be the \emph{$x$-distance} between~$v_i$ and~$v_j$.
To prove that a \good $st$-ordering $\mathcal{S}_x$ allows for a greedy 
realization, we develop a system of inequalities describing the geometric 
requirements for the $x$-distance of consecutive nodes in $\mathcal{S}_x$ in a 
greedy drawing, and then prove that this system always admits a solution 
since $\mathcal{S}_x$ is \good. First note that, for every $0<i<m$ such that 
there is no minimal $x$-conflict $\{v_i,v_{i+1}\}$, we only require the 
$x$-distance to be positive, so we define the following \emph{trivial inequality}.

\begin{equation*}
  x_{i,i+1}>0.\tag{trivial inequality}
  \label{eq:trivial}
\end{equation*}

For every $0<i<m$ such that there is a minimal $x$-conflict $\{v_i,v_{i+1}\}$,
we define two inequalities that describe the necessary conditions for the 
$x$-conflict to be \sat. Let~$u$ and $w$, with~$c_x(u)=v_i$ and~$c_x(w)=v_{i+1}$, be responsible 
for~$\{v_i,v_{i+1}\}$. We assume 
that $u \prec_y w$; the other case is symmetric.
\begin{figure}[t]
  \begin{minipage}[b]{.5\textwidth}
    \centering
    \includegraphics[page=1]{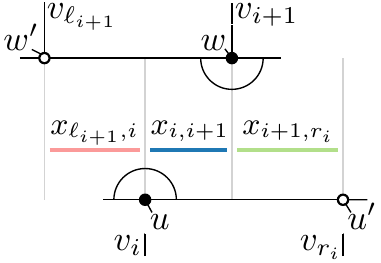}
    \hfill
    \includegraphics[page=2]{inequalities}
    \caption{Solving the inequalities of~$x_{i,i+1}$
      implies~$u\notin\cell(w)$ and~$w\notin\cell(u)$.}
    \label{fig:inequalities}
  \end{minipage}
  \hfill
  \begin{minipage}[b]{.46\textwidth}
    \centering
    \includegraphics{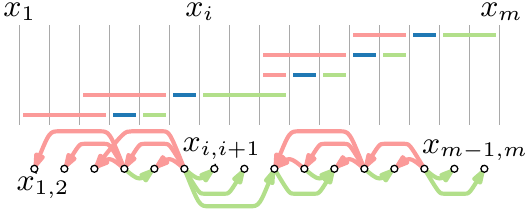}
    \caption{The relation graph defined by the left and right
      inequalities}
    \label{fig:relationgraph}
  \end{minipage}
\end{figure}

By assumption, $v_i$ lies to the bottom left
of~$v_{i+1}$, so we only have to consider the part of~$\cell(w)$ to the bottom left of~$w$, which we denote by $\cell_{\swarrow}(w)$ (dark region in Fig.~\ref{fig:inequalities}). 
Let~$(w',w)$ be the bottommost incoming edge of~$v_{i+1}$ with~$c_x(w')=v_{\ell_{i+1}}$. 
Then, the left boundary of~$\cell_{\swarrow}(w)$ is delimited by the vertical 
line through the mid-point of~$(w',w)$.
Thus, we require
\[x_{i,i+1}>x_{\ell_{i+1},i+1}/2\Leftrightarrow x_{i,i+1}>x_{{\ell_{i+1}},i}\]
Symmetrically, we only consider the part~$\cell_{\nearrow}(u)$ 
of~$\cell(u)$ to the top right of~$u$ (light region in 
Fig.~\ref{fig:inequalities}), which is bounded by the vertical line through the mid-point of the topmost outgoing edge~$(u,u')$ of~$v_i$ with~$c_x(u')=v_{r_i}$.
Thus, we require
\[x_{i,i+1}>x_{i,{r_i}}/2\Leftrightarrow x_{i,i+1}>x_{i+1,{r_i}}\]
Since $v_{\ell_{i+1}}$ and~$v_i$ (and $v_{i+1}$ and~$v_{r_i}$) are not 
necessarily consecutive in the $st$-ordering, we express the 
$x$-distance $x_{{\ell_{i+1}},i}$ (and~$x_{i,{r_i}}$) as the sum of 
the $x$-distances between the consecutive nodes between them 
in the $st$-ordering. This gives the \emph{left} and the 
\emph{right inequality}.

\begin{flalign*}
  x_{i,i+1}>\sum_{j={\ell_{i+1}}}^{i-1}x_{j,j+1} 
  & \quad\text{(left inequality)}
  & x_{i,i+1}>\sum_{j=i+1}^{{r_i}-1}x_{j,j+1}
  & \quad\text{(right inequality)}
\end{flalign*}

\noindent Note that for every variable~$x_{i,i+1}$ there exists either a trivial
inequality or a left and right inequality. Consider the following triangulated 
matrices, where $c_{i,j}=-1$ if $i>j\ge\ell_{i+1}$ or~$i<j\le r_i$,
where $c_{i,j}=1$ if~$i=j$, and~$c_{i,j}=0$ otherwise.

\begin{align*}
  A&=\left(\begin{matrix}
    c_{1,1} &  & 0 & \\
    \vdots & \ddots & \\
    c_{m-1,1} & \cdots & c_{m-1,m-1}
  \end{matrix}\right)\\
  B&=\left(\begin{matrix}
    c_{1,1} & \cdots & c_{1,m-1} \\
     & \ddots & \vdots& \\
    0 &  & c_{m-1,m-1}\\
  \end{matrix}\right)\\
  x&=\left(\begin{matrix}
  x_{1,2}\\
  \vdots\\
  x_{m-1,m}
  \end{matrix}\right)
\end{align*}
We express the left and trivial (right and trivial) inequalities as $Ax>0$ (as $Bx>0$).
Any vector~$x>0$ determines a unique rectilinear drawing:
we assign to each vertex the $y$-coordinate defined by $S_y$,
we assign to~$v_1$ the $x$-coordinate~$x_1=0$ and to every other~$v_i$ the $x$-coordinate $x_{i}=x_{i-1}+x_{i-1,i}$.
Since~$x>0$,
the $x$-coordinates preserve the \good $st$-ordering
and resolve all $x$-conflicts.

\begin{lemma}\label{lem:inequalities}
A vector~$x=(x_{1,2},\ldots,x_{m-1,m})^\top>0$ solves both~$Ax>0$ and $Bx>0$
if and only if it determines a drawing where all $x$-conflicts are \sat.
\end{lemma}
\begin{proof}
	First, suppose that~$x$ solves both~$Ax>0$ and~$Bx>0$.
	Let~$\{v_i,v_j\}$ be a minimal $x$-conflict. By Lemma~\ref{lem:consecutive},
	we have either~$j=i-1$ or~$j=i+1$; without loss of generality, assume $j=i+1$.
	Consider the $i$-th row~$a_i$ in matrix~$A$. By definition of~$A$, we have
	that~$a_i\cdot x>0$ is equivalent to the left inequality of~$\{v_i,v_j\}$.
	If the left inequality of~$\{v_i,v_j\}$ is \sat, then~$v_i$ lies outside
	the cell of~$v_j$; see Fig.~\ref{fig:inequalities}. 
	Analogously, the $i$-th row~$b_i$ in matrix~$B$
	gives the right inequality $b_i\cdot x>0$ of~$\{v_i,v_j\}$, which
	implies that~$v_j$ lies outside the cell of~$v_i$. Hence, the
	minimal $x$-conflict~$\{v_i,v_j\}$ is \sat
	and, by Lemmas~\ref{lem:conflict-sat-greedy} and~\ref{lem:dominations},
	all $x$-conflicts are resolved in~$\Gamma$.
	
	Now, suppose that~$x$ does not solve both~$Ax>0$ and~$Bx>0$; without loss 
	of generality, assume that some row~$a_i$ of~$A$ is not \sat, that is, we have
	\[x_{i,i+1}\le\sum_{j={\ell_{i+1}}}^{i-1}x_{j,j+1}=x_{\ell_{i+1},i}=x_{\ell_{i+1},i+1}-x_{i,i+1}.\]
	Then,~$v_i$ lies in $\cell(v_{i+1})$, so the drawing determined by~$x$
	is not greedy.
\end{proof}

Note that we can always solve $Ax>0$ and $Bx>0$ independently by solving the
linear equation systems $Ax=1$ and $Bx=1$ via forward %
substitution, since~$A$ and~$B$ are triangular. We prove that there is always a vector~$x$ solving~$Ax>0$ and~$Bx>0$ simultaneously.
Let~$C=A+B-I_{m-1}$ be the matrix defined by the values of~$c_{i,j}$.
We claim that any solution to the linear inequality system $Cx>0$ is also a solution to
both~$Ax>0$ and~$Bx>0$.
To see this, consider the inequalities described by the $i$-th row~$a_i$,~$b_i$, and $c_i$ of~$A$,~$B$, and~$C$, respectively. We have
\begin{eqnarray*}
a_i\cdot x>0&\Leftrightarrow& x_{i,i+1}>\sum_{j=1}^{i-1}(c_{i,j} x_{j,j+1}),\\
b_i\cdot x>0&\Leftrightarrow& x_{i,i+1}>\sum_{j=i+1}^{m-1}(c_{i,j} x_{j,j+1}),\\
c_i\cdot x>0&\Leftrightarrow& x_{i,i+1}>\sum_{j=1}^{i-1}(c_{i,j} x_{j,j+1}) + \sum_{j=i+1}^{m-1}(c_{i,j} x_{j,j+1}).
\end{eqnarray*}
So $c_i\cdot x>0$ implies both $a_i\cdot x>0$ and $b_i\cdot x>0$ and our claim follows. We now show that~$C$ can be triangulated.
For this, we define the \emph{relation graph} corresponding to the adjacency
matrix $I_{m-1}-C$ that contains a vertex~$u_i$
for each interval~$x_{i,i+1}$, $1\le i<m$, and a directed edge from a 
vertex~$u_i$ to a vertex~$u_j$ if and only if $c_{i,j}=-1$;
see Fig.~\ref{fig:relationgraph}.

\begin{lemma}\label{lem:acyclic}
  The relation graph of a \good $st$-ordering is acyclic.
\end{lemma}
\begin{proof}
	Let~$S_x=v_1,\ldots,v_m$ be a good $st$-ordering of~$D_x$,
	let~$A,B,C$ be the matrices as defined above,
	and let~$u_1,\ldots,u_m$ be the vertices of its relation graph.
	We call a directed edge~$(u_i,u_j)$ a \emph{left edge} if~$i<j$ and a
	\emph{right edge} otherwise. Note that a left (right) edge corresponds to a
	part of a left (right) inequality.
	We first have to prove the following property for the values $c_{i,j}$
	of the matrices~$A,B,C$.

	\begin{property}\label{prop:abconsecutive}
		For any~$0<i<j<k<m$, we have $c_{k,j}\le c_{k,i}$ and $c_{i,j}\le c_{i,k}$.
	\end{property}
	\begin{proof}
		Assume that $0=c_{k,j}>c_{k,i}=-1$. By definition of~$c_{k,i}$, we have 
		that $i\ge\ell_{k+1}$. But then $j>i\ge\ell_{k+1}$ implies~$c_{k,j}=-1$.
		Furthermore, assume that $0=c_{i,k}>c_{i,j}=-1$. By definition of~$c_{i,j}$,
		we have that~$k\le r_i$. But then $j<k\le r_i$ implies~$c_{i,k}=-1$.
	\end{proof}

	\noindent Consider the shortest cycle $u_{\lambda_1},\ldots,u_{\lambda_k},u_{\lambda_1}$ 
	in the relation graph. Obviously, there is at least one left edge and at least one right
	edge in the cycle. We will first show that this shortest cycle has length~$2$.
	Without loss of generality, assume that $(u_{\lambda_1},u_{\lambda_2})$ is a right edge; the
	other case is symmetric.
	Let~$i$ be the smallest number such that~$(u_{\lambda_i},u_{\lambda_{i+1}})$
	is a left edge. Then~$\lambda_1<\ldots<\lambda_i$, and we consider three
	cases.
	
	\medskip
	\newcase
	\renewcommand{\thecasecounter}{\Roman{casecounter}}
	\ccase{c:longcycle-equal} $\lambda_{i+1}=\lambda_{i-1}$.
	Then there is a cycle $u_{\lambda_{i-1}},u_{\lambda_{i}},u_{\lambda_{i+1}}=u_{\lambda_{i-1}}$
	of length~2.
	
	\smallskip
	\ccase{c:longcycle-smaller} $\lambda_{i+1}<\lambda_{i-1}<\lambda_i$.
	The edge~$(u_{\lambda_i},u_{\lambda_{i+1}})$ is a left edge, so 
	$c_{\lambda_{i},\lambda_{i+1}}=-1$. However, by Property~\ref{prop:abconsecutive},
	$c_{\lambda_{i},\lambda_{i-1}}\le c_{\lambda_{i},\lambda_{i+1}}=-1$,
	so there must be a left edge~$(u_{\lambda_i},u_{\lambda_{i-1}})$.
	Then there is a cycle $u_{\lambda_{i-1}},u_{\lambda_{i}},u_{\lambda_{i-1}}$
	of length~2.
	
	\smallskip
	\ccase{c:longcycle-larger} $\lambda_{i-1}<\lambda_{i+1}<\lambda_i$.
	The edge~$(u_{\lambda_{i-1}},u_{\lambda_{i}})$ is a right edge, so 
	$c_{\lambda_{i-1},\lambda_{i}}=-1$. However, by Property~\ref{prop:abconsecutive},
	$c_{\lambda_{i-1},\lambda_{i+1}}\le c_{\lambda_{i-1},\lambda_{i}}=-1$,
	so there must be a right edge~$(u_{\lambda_{i-1}},u_{\lambda_{i+1}})$.
	Hence, there is a shorter cycle 
	$u_{\lambda_1},\ldots,u_{\lambda_{i-1}},u_{\lambda_{i+1}},\ldots,u_{\lambda_{k}}$.
	\renewcommand{\thecasecounter}{\arabic{casecounter}}

	\medskip\noindent
	From our case analysis, it follows that~$k=2$. Let~$\alpha=\lambda_1$ and~$\beta=\lambda_2$.
	Then there are two minimal $x$-conflicts $\{\alpha,\alpha+1\}$ 
	and~$\{\beta,\beta+1\}$ with $\alpha<\beta$, $\alpha\ge\ell_{\beta+1}$, 
	and~$\beta+1\le r_{\alpha}$.
	Let $w_\alpha,w_{\alpha+1},w_\beta,w_{\beta+1}$ be the responsible
	vertices for these two $x$-conflicts with
	$c_x(w_\alpha)=v_\alpha$, $c_x(w_{\alpha+1})=v_{\alpha+1}$,
	$c_x(w_\beta)=v_\beta$, and $c_x(w_{\beta+1})=v_{\beta+1}$.
	Assume that $w_{\alpha}\prec_y w_{\alpha+1}$; the other case is symmetric.
	By the definition of $x$-conflicts, we have~$v_\alpha\not\prec_x v_\beta$ 
	and~$v_\beta\not\prec_x v_{\beta+1}$. 
	We first show that we cannot have~$\beta=\alpha+1$.
	
	\begin{figure}[t]
		\subcaptionbox{}{\includegraphics[page=1]{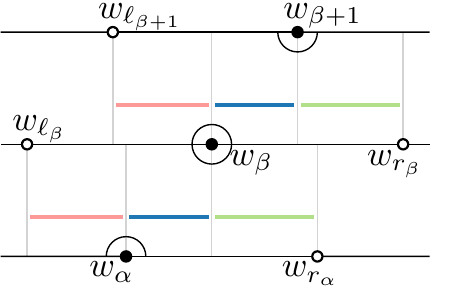}}
		\hfil
		\subcaptionbox{}{\includegraphics[page=2]{acyclicity-ba1}}
		\hfil
		\subcaptionbox{}{\includegraphics[page=3]{acyclicity-ba1}}
		\caption{Illustration for the proof of Property~\ref{prop:bnota1}: $\beta>\alpha+1$
			by assuming that~$\beta=\alpha+1$. (a) $v_\alpha\not\prec_x v_{\beta+1}$; (b) $v_\alpha\prec_x v_{\beta+1}$; (c) $v_\alpha\prec_x v_{\beta+1}$.}
		\label{fig:acyclicity-ba1}
	\end{figure}

	\begin{property}\label{prop:bnota1}
		$\beta>\alpha+1$.
	\end{property}
	\begin{proof} 
		Assume that $\beta = \alpha+1$; see Fig.~\ref{fig:acyclicity-ba1}.
		If there is no directed path between~$v_\alpha$ and~$v_{\beta+1}$,
		then the graph~$D_x\langle \alpha,\beta+1\rangle$
		contains three connected components, which contradicts Condition~\ref{c:strips-2comp}.  
		On the other hand, if there is a directed path between~$v_\alpha$ 
		and~$v_{\beta+1}$, then the graph~$D_x\langle \alpha,\beta+1\rangle$
		contains two connected components;
		one component that contains exactly~$v_\alpha$ and~$v_{\beta+1}$,
		and one component that contains only~$v_\beta$. However, since $\alpha<\beta<\beta+1$,
		this contradicts Condition~\ref{c:strips-precede}.
	\end{proof}

	\noindent We now show some properties on the 
	existence of directed paths between vertices~$v_\alpha,v_{\alpha+1},v_\beta,v_{\beta+1}$.
	
	\begin{property}\label{prop:path_a_b}
		$v_\alpha\prec_x v_\beta$. 
	\end{property}
	\begin{proof}
		Assume that~$v_\alpha\not\prec_x v_\beta$.
		Consider the graph~$D_x\langle \alpha,\beta+1\rangle$.
		If $v_\alpha\not\prec_x v_{\beta+1}$, then the graph has
		three connected components, which contradicts Condition~\ref{c:strips-2comp}. 
		Otherwise, there is a connected component
		that contains $v_\beta$ and a connected component that 
		contains~$v_\alpha$ and~$v_{\beta+1}$, which contradicts
		Condition~\ref{c:strips-precede}.
	\end{proof}
	
	\begin{property}\label{prop:path_a1_b1}
		$v_{\alpha+1}\prec_x v_{\beta+1}$.
	\end{property}
	\begin{proof}
		Assume that~$v_{\alpha+1}\not\prec_x v_{\beta+1}$. Consider the 
		graph~$D_x\langle \alpha,\beta+1\rangle$.
		If $v_\alpha\not\prec_x v_{\beta+1}$, then the graph has
		three connected components, which contradicts Condition~\ref{c:strips-2comp}. 
		Otherwise, there is a connected component
		that contains $v_{\alpha+1}$ and one that 
		contains~$v_\alpha$ and~$v_{\beta+1}$, which contradicts Condition~\ref{c:strips-precede}.
	\end{proof}
	
	\begin{property}\label{prop:path_a1_b} 
		$v_{\alpha+1}\prec_x v_{\beta}$.
	\end{property}
	\begin{proof}    
		Assume that~$v_{\alpha+1}\prec_x v_{\beta}$.
		Since $v_\alpha\prec_x v_\beta$, by 
		Property~\ref{prop:path_a_b}, it follows that
		the graph~$D_x\langle \alpha,\beta\rangle$ has a connected component
		that contains~$v_{\alpha+1}$ and one that 
		contains~$v_\alpha$ and~$v_{\beta}$; a contradiction to Condition~\ref{c:strips-precede}.
	\end{proof}

	Assume that we know the exact $y$-coordinates of every vertex. 
	We now have to analyze the relative positions of the vertices
	$w_\alpha,w_{\alpha+1},w_\beta,w_{\beta+1}$ with $y$-coordinates
	$\psi_\alpha,\psi_{\alpha+1},\psi_\beta,\psi_{\beta+1}$.  
	We will show
	that any choice of $y$-coordinates gives a contradiction. Note 
	that~$\psi_\alpha<\psi_{\alpha+1}$ by assumption and~$\psi_\beta\neq\psi_{\beta+1}$
	by the $x$-conflict~$\{\beta,\beta+1\}$.
	Recall that~$\alpha\ge\ell_{\beta+1}$ and~$\beta+1\le r_{\alpha}$.
	Further, let~$(w_\alpha,w_{r_{\alpha}})$ with $c_x(w_{r_{\alpha}})=v_{r_{\alpha}}$
	be the right horizontal edge of~$w_\alpha$,
	let $(w_{\ell_{\alpha+1}},w_{\alpha+1})$ with $c_x(w_{\ell_{\alpha+1}})=v_{\ell_{\alpha+1}}$
	be the left horizontal edge of~$w_{\alpha+1}$, 
	let $(w_\beta,w_{r_{\beta}}$ with $c_x(w_{r_{\beta}})=v_{r_{\beta}}$
	be the right horizontal edge of~$w_\beta$,
	and let $(w_{\ell_{\beta+1}},w_{\beta+1})$ with $c_x(w_{\ell_{\beta+1}})=v_{\ell_{\beta+1}}$
	be the left horizontal edge of~$w_{\beta+1}$; refer to the definition of the
	left and right inequalities.
	We distinguish between the following cases. 
	
	\begin{figure}[t]
		\centering
		\subcaptionbox{\label{fig:acyclicity-bb1-a1b1-1}}{\includegraphics[page=1]{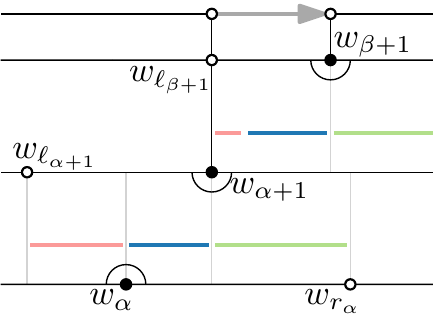}}
		\hfill
		\subcaptionbox{\label{fig:acyclicity-bb1-a1b1-2}}{\includegraphics[page=2]{acyclicity-bb1}}
		\hfill
		\subcaptionbox{\label{fig:acyclicity-bb1-b1a1}}{\includegraphics[page=3]{acyclicity-bb1}}
		\caption{Illustration for Case~\ref{sc:nocycle-neq-bb1}:  $\psi_\beta<\psi_{\beta+1}$;
			(a)--(b) The two cases for the path from~$v_{\alpha+1}$ to~$v_{\beta+1}$
			in  Case~\ref{ssc:nocycle-neq-bb1-b1_above_a1}: $\psi_{\beta+1}>\psi_{\alpha+1}$.
			in (a)~the  path starts above~$\psi_{\beta+1}$, while in (b)~the path starts below~$\psi_{\beta+1}$.
			(c) Case~\ref{ssc:nocycle-neq-bb1-b1_below_a1}: $\psi_{\beta}<\psi_{\beta+1}<\psi_{\alpha+1}$. 
			The path from~$v_{\alpha+1}$ to~$v_\beta$.}
		\label{fig:acyclicity-bb1}
	\end{figure}
		
	\newcase
	\medskip
	\ccase{sc:nocycle-neq-bb1}
	$\psi_\beta<\psi_{\beta+1}$. This implies that~$w_\beta\prec_y w_{\beta+1}$
	and thus~$w_\beta$ has a north-oriented flat angle and~$w_{\beta+1}$ has a 
	south-oriented flat angle; see Fig.~\ref{fig:acyclicity-bb1}.
	
	\smallskip
	\subcase{ssc:nocycle-neq-bb1-b1_above_a1}
	$\psi_{\beta+1}>\psi_{\alpha+1}$.
	By Property~\ref{prop:path_a1_b1},~$v_{\alpha+1}\prec_x v_{\beta+1}$. 
	If the corresponding path starts in~$v_{\alpha+1}$ at a $y$-coordinate $\ge\psi_{\beta+1}$,
	then we have that $\alpha\ge\ell_{\beta+1}\ge\alpha+1$, a contradiction;
	see Fig.~\ref{fig:acyclicity-bb1-a1b1-1}.
	Otherwise, since~$v_{\beta+1}$ has a south-oriented flat angle,
	this path has to end at a $y$-coordinate $\ge\psi_{\beta+1}$
	and its last segment is a horizontal segment.
	Hence, the path has to traverse some point with 
	$y$-coordinate~$\psi_{\beta+1}$ and with $x$-coordinate between~$\alpha+1$
	and~$\beta+1$; see Fig.~\ref{fig:acyclicity-bb1-a1b1-2}. 
	However, all of these points lie on the 
	edge~$(w_{\ell_{\beta+1}},w_{\beta+1})$, due to~$\ell_{\beta+1}\le\alpha$,
	which contradicts planarity.
	
	\smallskip
	\subcase{ssc:nocycle-neq-bb1-b1_at_a1}
	$\psi_{\beta+1}=\psi_{\alpha+1}$.
	Then we have $\ell_{\beta+1}=\alpha+1>\ell_{\beta+1}$; a contradiction.
	
	\smallskip
	\subcase{ssc:nocycle-neq-bb1-b1_below_a1}
	$\psi_{\beta}<\psi_{\beta+1}<\psi_{\alpha+1}$; see Fig.~\ref{fig:acyclicity-bb1-b1a1}.
	By Property~\ref{prop:path_a1_b},~$v_{\alpha+1}\prec_x v_{\beta}$. 
	Since~$w_{\alpha+1}$ has a south-oriented flat angle
	and~$w_{\beta}$ has a north-oriented flat angle, the corresponding path has to
	traverse some point with $y$-coordinate~$\psi_{\beta+1}$ 
	and with $x$-coordinate between~$v_{\alpha+1}$ and~$v_{\beta}$.
	However, all of these points lie on the edge $(w_{\ell_{\beta+1}},w_{\beta+1})$,
	due to~$\ell_{\beta+1}\le\alpha$, which contradicts planarity.

	\medskip
	\ccase{sc:nocycle-neq-b1b}
	$\psi_{\beta+1}<\psi_{\beta}$. This implies that~$w_{\beta+1}\prec_y w_{\beta}$
	and thus~$w_{\beta+1}$ has a north-oriented flat angle and~$v_{\beta}$ has a 
	south-oriented flat angle; see Fig.~\ref{fig:acyclicity-b1b}.
	
	\smallskip
	\subcase{ssc:nocycle-neq-b1b-b1_above_a1}
	$\psi_\beta>\psi_{\beta+1}>\psi_{\alpha+1}$.
	By Property~\ref{prop:path_a1_b},~$v_{\alpha+1}\prec_x v_{\beta}$. 
	If the corresponding path starts in~$v_{\alpha+1}$ at a $y$-coordinate $>\psi_{\beta+1}$,
	then we have that $\alpha\ge\ell_{\beta+1}\ge\alpha+1$, a contradiction;
	see Fig.~\ref{fig:acyclicity-b1b-a1b1-1}.
	Otherwise, since~$w_{\beta}$ has a south-oriented flat angle,
	this path has to end at a $y$-coordinate $\ge\psi_{\beta}>\psi_{\beta+1}$
	and its last segment is a horizontal segment.
	Hence, the path has to traverse some point with 
	$y$-coordinate~$\psi_{\beta+1}$ and with $x$-coordinate between~$\alpha+1$
	and~$\beta$; see Fig.~\ref{fig:acyclicity-b1b-a1b1-2}. 
	However, all of these points lie on the 
	edge~$(w_{\ell_{\beta+1}},w_{\beta+1})$, due to~$\ell_{\beta+1}\le\alpha$,
	which contradicts planarity. 
	
	\smallskip
	\subcase{ssc:nocycle-neq-b1b-b1_at_a1}
	$\psi_{\beta+1}=\psi_{\alpha+1}$.
	Then we have $\ell_{\beta+1}=\alpha+1>\ell_{\beta+1}$; a contradiction.
	
	\smallskip
	\subcase{ssc:nocycle-neq-b1b-b1_below_a1}
	$\psi_{\beta+1}<\psi_{\alpha+1}$; see Fig.~\ref{fig:acyclicity-b1b-b1a1}.
	By Property~\ref{prop:path_a1_b1},~$v_{\alpha+1}\prec_x~v_{\beta+1}$. 
	Since~$w_{\alpha+1}$ has a south-oriented flat angle
	and~$w_{\beta+1}$ has a north-oriented flat angle, the corresponding path has 
	to traverse some point with $y$-coordinate~$\psi_{\beta+1}$ 
	and with $x$-coordinate between~$v_{\alpha+1}$ and~$v_{\beta+1}$.
	However, all of these points lie on the edge $(w_{\ell_{\beta+1}},w_{\beta+1})$,
	due to~$\ell_{\beta+1}\le\alpha$, which contradicts planarity.
	
	\medskip From the above case analysis, it follows that there is no valid $y$-coordinate~$\psi_{\beta+1}$
	in any rectilinear drawing. Thus, there cannot be any cycle in the
	relation graph and the proof of the lemma follows.
\end{proof}

	\begin{figure}[t]
		\centering
		\subcaptionbox{\label{fig:acyclicity-b1b-a1b1-1}}{\includegraphics[page=1]{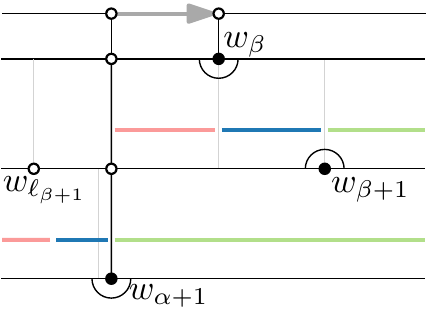}}
		\hfill
		\subcaptionbox{\label{fig:acyclicity-b1b-a1b1-2}}{\includegraphics[page=2]{acyclicity-b1b}}
		\hfill
		\subcaptionbox{\label{fig:acyclicity-b1b-b1a1}}{\includegraphics[page=3]{acyclicity-b1b}}
		\caption{Illustration for Case~\ref{sc:nocycle-neq-b1b}:  $\psi_{\beta+1}<\psi_{\beta}$. 
			(a)--(b) The two cases for the path from~$v_{\alpha+1}$ to~$v_{\beta}$
			in  Case~\ref{ssc:nocycle-neq-b1b-b1_above_a1}: $\psi_\beta>\psi_{\beta+1}>\psi_{\alpha+1}$;
			in (a)~the path starts above~$\psi_{\beta+1}$, while in (b)~the path starts below~$\psi_{\beta+1}$.
			(c) Case~\ref{ssc:nocycle-neq-b1b-b1_below_a1}: $\psi_{\beta+1}<\psi_{\alpha+1}$.
			The path from~$v_{\alpha+1}$ to~$v_{\beta+1}$.}
		\label{fig:acyclicity-b1b}
	\end{figure}

From the acyclicity of the relation graph, we show in the following lemma that~$C$ is
triangularizable.

\begin{lemma}\label{lem:triangularizable}
  The matrix~$C$ is triangularizable.
\end{lemma}
\begin{proof}
By Lemma~\ref{lem:acyclic}, the relation graph described by the matrix~$I_{m-1}-C$ is acyclic. Hence, there is a permutation matrix~$P$
(corresponding to a topological sort) such that~$P(I_{m-1}-C)P^{-1}$
is triangulated with only $0$'s on the diagonal. Thus,~$PI_{m-1}P^{-1}-PCP^{-1}=I_{m-1}-PCP^{-1}$ is triangulated with only $0$'s on the diagonal, so $PCP^{-1}$ is triangulated with only $1$'s on the diagonal.
\end{proof}

Since~$C$ is triangularizable by Lemma~\ref{lem:triangularizable}, the system
of linear equations $Cx=1$ always has a solution, which solves~$Ax>0$ and~$Bx>0$ simultaneously.
This concludes the sufficiency proof for Theorem~\ref{thm:characterization}.

\subsection{Area requirements of greedy rectilinear drawings}
\medskip
In this subsection we consider the area requirements of greedy drawings of rectilinear representations. The first observation in this direction is that the construction presented in the previous subsection ensures that all the coordinates are integer. However, the area of the produced drawing is in general not minimum, since we just require an interval to be larger than the sum of the length of all intervals of the left and right inequality. 

In the following, we strengthen the algorithmic part of the characterization by showing that, given \good $st$-orderings of the two DAGs $D_x$ and~$D_y$, we can construct in polynomial time a rectilinear greedy drawing with minimum area respecting the given $st$-orderings. 

\begin{theorem}\label{thm:construction}
  Let $H$ be a \convex rectilinear representation of a biconnected  
  plane graph and let~$\mathcal S_x$ and~$\mathcal S_y$
  be \good $st$-orderings of~$D_x$ and~$D_y$. We can compute a
  greedy drawing of~$H$ that respects~$\mathcal S_x$
  and~$\mathcal S_y$ with minimum area in $O(n^2)$ time.
\end{theorem}
\begin{proof}
	By Theorem~\ref{thm:characterization}, there is always a rectilinear greedy 
	drawing of~$H$, and we can construct one by solving the linear equality 
	system~$Cx=1$ as described above for both~$D_x$ and~$D_y$. Since all 
	inequalities are necessary and sufficient, by Lemma~\ref{lem:inequalities}, a 
	solution of minimum area will have the form $x_{i,i+1}=1$, if there is no 
	minimal $x$-conflict~$\{v_i,v_{i+1}\}$,
	while otherwise we have:
	\[x_{i,i+1}=\max\{\sum_{j=i+1}^{{r_i}-1}x_{j,j+1},\sum_{j={\ell_{i+1}}}^{i-1}x_{j,j+1}\}+1\]
	We can find such a solution in quadratic time by solving the following linear program.
	\[
	\text{minimize}\hspace{1em} \sum_{i=1}^{m-1} x_{i,i+1} \hspace{2em}  
	\text{subject to}\hspace{1em} Ax\ge 1 \hspace{1em}\text{and}\hspace{1em} Bx\ge 1.
	\]
	
	Note that the inequalities already imply~$x\ge 1$.
	By the acyclicity of the constraints, there is always a
	solution that satisfies~$a_ix_i=1$ or~$b_ix_i=1$ for each~$1\le i<m$,
	where~$a_i$ and~$b_i$ correspond to the $i$-th row of the matrices~$A$
	and~$B$, respectively;
	hence, the linear program will assign to each~$x_{i,i+1}$ the value
	$x_{i,i+1}=\max\{\sum_{j=i+1}^{{r_i}-1}x_{j,j+1},\sum_{j={\ell_{i+1}}}^{i-1}x_{j,j+1}\}+1$, 
	which is an integer\footnote{Formally, one would have to prove that the constraint
		matrix is totally unimodular, from which we refrain here since the fact that
		we obtain an integral solution should be clear.}.
	
	For the running time, we first have to find all minimal $x$-conflicts. To this
	end, we only have to check whether two consecutive nodes in the
	$st$-orderings have an $x$-conflict; this can clearly be done in linear time
	per node pair, so in~$O(n^2)$ time in total. Then, we have to 
	create the matrices~$A$,~$B$, and~$C$, which have at most~$m-1$ rows and
	columns each (since~$D_x$ and~$D_y$ might have fewer nodes than~$H$). This takes~$O(n^2)$ time each. In order to triangularize~$C$, we have to compute a
	topological order on the DAG defined by the adjacency matrix~$I_m-C$;
	this can be done in~$O(n^2)$ time using, e.g., depth-first search.
	Finally, we can solve the linear program in polynomial time. 
	
	In general, it is not known whether the linear program can be solved
	in~$O(n^2)$ time; the best-known bound is~$O^*(n^{\omega})$
	where~$\omega\approx 2.372$ is the current matrix multiplication time~\cite{cls-slpcm-stoc19}.
	However, we can reduce the runtime for finding a rectilinear
	greedy drawing of~$H$ with minimum area by solving the inequalities ``by hand''.
	Let~$A^*=\left(a_1^*,\ldots,a_m^*\right)^\top=PAP^{-1}$
	and~$B^*=\left(b_1^*,\ldots,b_m^*\right)^\top=PBP^{-1}$.
	Let~$C^*=\left(a_1^*,b_1^*,\ldots,a_m^*,b_m^*\right)^\top$
	and $x^*=xP^{-1}$.
	Obviously, the following linear program
	is equivalent to the one above and since both~$A^*$ and~$B^*$ are upper
	triangulated, we can solve it bottom-up two rows at a time in $O(n^2)$ time.
	\[
	\text{minimize}\hspace{1em} \sum_{i=1}^{m-1} x_{i,i+1} \hspace{2em}
	\text{subject to}\hspace{1em} C^*x^*\ge 1
	\]

	We can also use a more algorithmical approach.
	We can assign the values to each~$x_{i,i+1}$ already while using the
	topological sort to triangulate~$I_m-C$; according to this topological
	sort, we can assign~$x_{i,i+1}=1$ to all sources of the DAG $I_m-C$
	and the maximum of $\sum_{j=i+1}^{{r_i}-1}x_{j,j+1}+1$ and
	$\sum_{j={\ell_{i+1}}}^{i-1}x_{j,j+1}+1$ to all non-sources.
	By this, all linear inequalities are \sat and the minimality follows
	by the necessity of the constraints. Since the DAG~$I_m-C$ has at most~$O(m^2)$
	edges, this algorithm works in~$O(n^2)$ time.
\end{proof}

In the following we show that, although minimum, the area of the drawings produced by our algorithm may be non-polynomial in some cases; namely, Theorem~\ref{th:exponential-area} states that there exist \convex rectilinear representations whose DAGs admit \good $st$-orderings, but there is no combination of them resulting in a succinct greedy drawing, since the solutions of the corresponding system of inequalities are always exponential in the input size. Observe that, on the contrary, every universal greedy rectilinear representation of an $n$-vertex graph is succinct, since by Corollary~\ref{co:universal-min-area} it has a (greedy) drawing of minimum area on an integer grid of size $O(n^2$)~\cite{DBLP:journals/comgeo/BridgemanBDLTV00,DBLP:journals/siamcomp/Tamassia87}.

\begin{theorem}\label{th:exponential-area}
There exist rectilinear representations whose every greedy rectilinear drawing has exponential area, even if~$D_x$ and~$D_y$ are series-parallel.
\end{theorem}
\begin{proof}
	We first describe a rectilinear representation $H$, and then we show that it 
	satisfies the properties of the statement; see Fig.~\ref{fig:exp-staircase}. 
	The vertex set of $H$ consists of four sets $v_1,\dots,v_q$, $w_1,\dots,w_{q-1}$, 
	$z_1,\dots,z_q$, and $u_2,\dots,u_q$, connected as follows. Vertices $v_1,
	\dots,v_q$ belong to a vertical path $\pi_x(v)$, so that they appear in this 
	order from bottom to top. Then, for each $i = 2, \dots, q-1$, we add a 
	horizontal path $\pi_y^i = u_i,w_i,z_i,v_i$ such that these vertices appear in 
	this left-to-right order. Also, we add a horizontal path 
	$\pi_y^1 = w_1,z_1,v_1$ and a horizontal path~$\pi_y^q=u_q,z_q,v_q$
	such that these vertices appear in this left-to-right order. Finally, for 
	each $i = 1, \dots, q-1$, we add a vertical path~$\pi_x^i$ composed of a 
	single edge $(w_i,u_{i+1})$.
	Observe that $H$ is \convex.
	
	\begin{figure}[t]
		\centering
		\subcaptionbox{\label{fig:exp-staircase}}{\includegraphics[page=1]{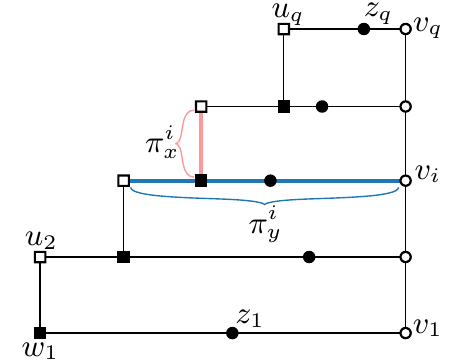}}
		\hfill
		\subcaptionbox{\label{fig:exp-seriesparallel}}{\includegraphics[page=2]{staircase}}
		\hfill
		\subcaptionbox{\label{fig:exp-constraints}}{\includegraphics[page=3]{staircase}}
		\caption{Illustration for the proof of Theorem~\ref{th:exponential-area}. (a) Rectilinear representation~$H$; (b) The DAG~$D_x$; (c) A \good $st$-ordering and its constraints.}
		\label{fig:exponential}
	\end{figure}
	
	We now consider DAGs~$D_x$ and $D_y$, and their possible \good $st$-orderings. 
	The first observation is that~$D_y$ is a directed path from the vertex 
	$c_y(v_1)$ corresponding to $\pi_y^1$ to the vertex $c_y(v_q)$ corresponding 
	to~$\pi_y^q$. Thus, $D_y$ admits a unique $st$-ordering, which is trivially 
	\good. As for $D_x$, it consists of the series-parallel graph depicted in 
	Fig.~\ref{fig:exp-seriesparallel}, whose unique source is the vertex $c_x(w_1)$ 
	corresponding to $\pi_x^1$ and whose unique sink is the vertex $c_x(v_1)$ 
	corresponding to $\pi_x(v)$. Although $D_x$ admits several st-orderings, we 
	claim that only two of them are \good.
	
	Observe that $D_x$ contains a directed path 
	$c_x(w_1), \dots,c_x(w_{q-1}),c_x(z_q),c_x(v_1)$, and thus these vertices appear
	in this order in any st-ordering. Thus, the only possible $st$-orderings differ 
	by the placement of 
	vertices $c_x(z_1), \dots, c_x(z_{q-1})$. Note that $c_x(z_{q-1})$ must 
	appear after $c_x(w_{q-1})$ in any $st$-ordering. Thus, if we consider the 
	subgraph of $D_x$ induced by the vertices from the one following~$c_x(w_{q-1})$
  to the one preceding $c_x(v_1)$ in any $st$-ordering, we always have a 
	connected component consisting only of $c_x(z_{q-1})$, and another connected 
	component consisting only of $c_x(z_{q})$. This implies that 
	no other vertex can be placed after $c_x(w_{q-1})$, as otherwise the 
	resulting $st$-ordering would not be \good. 
	Consider now the subgraph of $D_x$ induced by the vertices from the one 
	following~$c_x(w_{q-2})$ to the one preceding $c_x(v_1)$ in any $st$-
	ordering. Again, we have already two connected components, namely one 
	consisting only of $c_x(z_{q-2})$ and one consisting 
	of $c_x(w_{q-1})$, $c_x(z_{q-1})$, and $c_x(z_{q})$. This implies that no other vertex can 
	be placed after $c_x(w_{q-2})$, as otherwise the resulting $st$-ordering 
	would not be \good. In particular, this implies that $c_x(w_{q-2})$, $c_x(z_{q
		-2})$, and~$c_x(w_{q-1})$ are consecutive in any \good $st$-ordering of $D_x$.
	Repeating this argument for every $i \leq q-2$, we obtain that $c_x(w_1), c_x(z_1), c_x(w_2),c_x(z_2), \dots, c_x(w_{q-2}), c_x(z_{q-2})$, and $c_x(w_{q-1})$ are consecutive in any \good $st$-ordering of $D_x$. Thus, there exist only two \good $st$-orderings, which only differ for the position of $c_x(z_{q-1})$ with respect to the positions of $c_x(w_{q})$ and of~$c_x(z_{q})$; recall that, in a \good $st$-ordering, $c_x(z_{q-1})$ must appear either before or after both of $c_x(w_{q})$ and $c_x(z_{q})$.
	
	Assume that $c_x(z_{q-1})$ appears before~$c_x(z_{q})$ in the \good $st$-ordering;
	see Fig.~\ref{fig:exp-constraints}. The~other case is analogous.
	For ease of notation, we say~$w_q:=z_q$.
	By the good $st$-ordering, we have:
	\[
	x(z_i)>x(u_i)
	\label{eq:exp-1}
	\]
	for every $1\le i<q$.
	Recall the definitions of
	minimal $x$-conflict and right inequality from Section~\ref{se:greedy-rectilinear}.
	For each~$1\le i<q$, we have
	a minimal $x$-conflict between~$c_x(z_i)$ and~$c_x(w_{i+1})$
	with $r_{c_x(z_i)}=c_x(v_1)$.
	This gives us the right inequality:
	
	\begin{align}
	&x(w_{i+1})-x(z_{i})>x(v_1)-x(w_{i+1}) \notag\\
	\Leftrightarrow\quad&  x(v_1)-x(z_i)>2\left(x(v_1)-x(w_{i+1})\right)
	>_{(\ref{eq:exp-1})} 2\left(x(v_1)-x(z_{i+1})\right)
	\label{eq:exp-2}
	\end{align}
	
	Hence, we obtain:
	
	\begin{align}
	x(v_1)-x(z_1)>_{(\ref{eq:exp-2})}2\left(x(v_1)-x(z_2)\right)>\ldots
	>_{(\ref{eq:exp-2})}2^{q-1}\left(x(v_1)-x(z_q)\right).
	\label{eq:exp-3}
	\end{align}
	
	Since a solution to the right inequalities is necessary for a
	rectilinear greedy drawing by Lemma~\ref{lem:inequalities},
	any greedy drawing of~$H$ must satisfy Equation~\ref{eq:exp-3}.
	However, by $q\in \Omega(n)$, this implies a lower bound on the area
	of~$2^{\Omega(n)}$.
\end{proof}

\subsection{A linear-time algorithm for a special family of instances}
\medskip
We conclude the section by presenting an algorithm to efficiently test, for a meaningful subset of instances, whether a rectilinear representation is greedy realizable. In particular we show that, when an $st$-graph is series-parallel, it is possible to test efficiently whether it admits a \good $st$-ordering, and thus satisfies the condition of the characterization presented in Theorem~\ref{thm:characterization}.

\begin{theorem}\label{thm:seriesparallel}
Let $H$ be a \convex rectilinear representation of a biconnected plane 
graph. If~$D_x$ and~$D_y$ are series-parallel, we can 
test in $O(n)$ time if $H$ is greedy realizable. If the test 
succeeds, a greedy drawing of $H$ is computed in~$O(n^2)$ time.
\end{theorem}
\begin{proof}
	By Theorem~\ref{thm:characterization}, we need to check whether both $D_x$ 
	and $D_y$ admit a \good $st$-ordering. We show how to check this for $D_x$ in 
	linear time, the algorithm for~$D_y$ is the same. 
	
	Consider the recursive construction of $D_x$ through series and parallel compositions. For the base case, notice that a graph consisting of a single edge trivially has a \good $st$-ordering. 
	Let $D_x$ be composed of a set of subgraphs $D_1,\dots, D_k$, forming a 
	parallel or a series composition. If we assume that if $D_x$ is composed by a 
	parallel (resp. series) composition, then each of $D_i$ was composed by series 
	(resp. parallel) composition. A construction with this property can be 
	obtained by considering each composition to be maximal.
	
	First note that, if $D_1,\dots, D_k$ form a parallel composition, then either $k=2$
	or $k=3$ and one of~$D_1,D_2,D_3$ is a single edge. 
	In fact, let $s$ and $t$ be the source and sink of $D_1,\dots, D_k$. Thus, for any 
	$st$-ordering $\mathcal{S}_x=v_1,\dots,v_m$ of $D_x$, it holds that $s = v_1$, 
	$t=v_m$, and for each internal vertex $u$ of a component in $D_1,\dots, D_k$, we 
	have $u = v_q$, for some $1 < q < m$. This implies that, for each component 
	in $D_1,\dots, D_k$ that is not a single edge, there exists a connected component 
	in $D_x\langle 1,m \rangle$. Hence, both $k > 3$ and $k=3$ where none of
	$D_1,D_2,D_3$ is a single edge would violate 
	Condition~\ref{c:strips-2comp} of a \good $st$-ordering. 
	
	Consider now a parallel composition between two vertices $s$ and $t$ 
	consisting of exactly two components $D_1$ and $D_2$ that are not a single edge. 
	Let $D^{\core}$ denote 
	the subgraph of $D_x$ induced by the nodes $V(D_x) \setminus \{s,t\}$. Recall 
	that, by Condition~\ref{c:strips-precede}, all nodes of $D_1^\core$ must 
	precede all nodes of $D_2^\core$ in a \good $st$-ordering, or vice versa. 
	Consider the case in which all nodes of $D_1^\core$ precede those of 
	$D_2^\core$, the other one is analogous. We claim that this results in a \good 
	$st$-ordering only if $D_1^\core$ has a single sink and $D_2^\core$ has a single 
	source. This follows from the observation that, for any set of sinks of 
	$D_1^\core$ and sources of $D_2^\core$, it is possible to find a pair of 
	nodes $v_p$ and~$v_q$ in any $st$-ordering $\mathcal{S}_x=v_1,\dots,v_m$ such that each 
	of these sources/sinks define a connected component 
	in $D_x\langle p,q \rangle$; thus, if there exist more than two sources/sinks, 
	then there exists no 
	\good $st$-ordering. On the other hand, if $D_1^\core$ has only one sink 
	and $D_2^\core$ only one source, none of the conditions for a \good $st$-ordering 
	are violated. From the above discussion, it follows that the only two checks 
	to perform are whether either~$D_1^\core$ has only one sink and $D_2^\core$ 
	only one source, and vice versa. If one of the checks succeeds, we compute a 
	\good $st$-ordering of $D_1^\core$ and of $D_2^\core$, and we merge them 
	according to the result of the check; otherwise, we reject the instance.
	
	When $D_1,\dots, D_k$ form a series composition, the number of 
	components $D_1,\dots, D_k$ and their structure can be arbitrary. We 
	construct \good $st$-orderings of $D_1,\dots, D_k$ recursively and merge them in 
	a \good  $st$-ordering of $D_x$.
	
	To conclude, the necessary and sufficient condition for $D_x$ to have a \good 
	$st$-ordering is that at every parallel composition either exactly two components 
	are merged or exactly three components are merged and additionally one of them is a single 
	edge, one has a single source and one has a single sink. 
	This condition can be checked in time linear to the number of 
	nodes of $D_x$. The time complexity for the construction of a rectilinear 
	greedy drawing follows from Theorem~\ref{thm:construction}.
\end{proof}

\section{Conclusions and Open Problems}\label{se:conclusions}

In this work, we introduced rectilinear greedy drawings, i.e., planar greedy drawings in the orthogonal drawing style with no bends. Our work reveals several interesting open problems.
\begin{enumerate}
\item The main problem raised by our work is whether we can test in polynomial time whether whether a rectilinear representation is greedy realizable. Due to our characterization, this is equivalent to asking whether a planar DAG admits a \good $st$-ordering. 

\item For the aforementioned open problem, we provided a linear-time testing algorithm when the DAG is series-parallel. As a further step towards an answer to our main open problem, it is worth studying the special case in which the DAG has only one source and one sink.

\item It is known that not all degree-$4$ plane graphs admit a rectilinear representation, while all of them have an orthogonal representation with bends~\cite{DBLP:journals/siamcomp/Tamassia87}. This motivates to extend the study to greedy orthogonal drawings with bends along the edges.

\item Given a biconnected plane graph $G$ (that is, without prescribed values for the geometric angles around each vertex), what is the complexity of deciding whether~$G$ admits a (universal) greedy rectilinear representation? This question pertains the intermediate step of the topology-shape-metrics approach~\cite{DBLP:journals/siamcomp/Tamassia87}.
\end{enumerate}

\bibliography{abbrv,GreedyOrtho}
\bibliographystyle{abbrvurl}

\end{document}